\documentclass[aps,prx,preprint,groupedaddress,showkeys]{revtex4-2}
\usepackage{amsmath,amssymb,amsthm,bm,mathtools}
\usepackage{graphicx,booktabs,placeins,microtype}
\usepackage{needspace}
\usepackage[hidelinks]{hyperref}
\hypersetup{pdftitle={A deterministic polynomial chaos Galerkin method for non-Markovian quantum state diffusion},pdfauthor={Quanhui Zhu and Zhenning Cai}}

\newcommand{\dd}{\mathrm{d}}
\newcommand{\E}{\mathbb{E}}
\newcommand{\Tr}{\operatorname{Tr}}
\newcommand{\CN}{\mathcal{CN}}
\newcommand{\norm}[1]{\left\lVert #1\right\rVert}
\makeatletter
\let\inputtablerows\@@input
\makeatother

\newtheorem{assumption}{Assumption}[section]

\newtheorem{theorem}[assumption]{Theorem}

\theoremstyle{remark}

\theoremstyle{definition}
\newtheorem{example}{Example}[section]
\AddToHook{env/example/before}{\Needspace{8\baselineskip}}

\begin{document}
\title{A deterministic polynomial chaos Galerkin method for non-Markovian quantum state diffusion}
\author{Quanhui Zhu}
\email[Contact author: ]{zhuqh@nus.edu.sg}
\author{Zhenning Cai}
\email[Contact author: ]{matcz@nus.edu.sg}
\affiliation{Department of Mathematics, National University of Singapore, Singapore}
\date{\today}

\begin{abstract}
Developing general-purpose numerical methods for non-Markovian quantum
state diffusion has long been a challenge in the simulation of open
quantum systems.
This paper proposes a deterministic polynomial chaos Galerkin method,
which is easy to implement, accurate, and efficient.
By combining a low-rank decomposition of bath correlations with
Galerkin projection onto a polynomial chaos basis, the method yields
a deterministic formulation compatible with standard numerical solvers
for partial differential equations. The reduced density matrix is
obtained from a single deterministic solve at a cost comparable to that
of computing one stochastic hierarchy trajectory. Replacing an ensemble of
$N_{\mathrm{traj}}$ trajectories therefore reduces the computational
cost by approximately a factor of $N_{\mathrm{traj}}$ and eliminates
statistical sampling error.
Benchmark comparisons show that the proposed method achieves smaller
errors at lower computational cost than the tested stochastic hierarchy methods.
Applications involving nonquadratic potentials, diffraction and interference,
and coupled three-dimensional dynamics demonstrate its ability to handle
challenging problems.
\end{abstract}

\maketitle

\section{Introduction}
\label{sec:introduction}

Open quantum systems interact with an environment, and a central
computational problem is to determine the reduced dynamics of the system
without explicitly simulating the full system--environment
evolution~\cite{breuer_theoryopenquantum_2002,RevModPhys.88.021002,devega_dynamicsnonmarkovianopen_2017}.
Accurate simulation of these dynamics is challenging
in the non-Markovian regime, where bath memory cannot be neglected,
particularly for systems with many degrees of freedom~\cite{makri_smallmatrixdisentanglement_2020,keeling_processtensorapproaches_2026}.
Established approaches range from influence-functional methods~\cite{feynman1963},
including the quasiadiabatic propagator path integral~\cite{makri_tensorpropagatoriterative_1995,makri1995numerical}
and time-evolving matrix product operators~\cite{strathearn_efficientnonmarkovianquantum_2018a,jorgensen2019},
to the hierarchical equations of motion (HEOM)~\cite{tanimura1989,tanimura_numericallyexactapproach_2020}.
More recently, inchworm resummation~\cite{wang2025} and low-rank Dyson
methods~\cite{zhan2026} have been used to reduce the computational cost of
non-Markovian simulations.

Stochastic Schr\"odinger equations offer a complementary approach by
unravelling the reduced dynamics into an ensemble of pure-state
trajectories~\cite{dalibard_wavefunctionapproachdissipative_1992,gisin_quantumstatediffusionmodel_1992}.
The reduced density operator is recovered by averaging their
wavefunction outer products, so each realization propagates a state
vector with fewer dynamical variables than a density matrix.
Non-Markovian quantum state diffusion (NMQSD) incorporates bath memory
into this description through a colored Gaussian process and a
functional-derivative term~\cite{diosi_nonmarkovianstochasticschrodinger_1997,diosi_nonmarkovianquantumstate_1998a,strunz1999}.
This functional derivative makes direct time propagation difficult.
The development of a general-purpose numerical algorithm was explicitly
posed as an open problem in the original NMQSD
paper~\cite{diosi_nonmarkovianquantumstate_1998a}.

One approach to handling the functional derivative is the hierarchy
of pure states (HOPS), which introduces auxiliary wavefunctions that satisfy
coupled evolution
equations~\cite{suess_hierarchystochasticpure_2014,hartmann_exactopenquantum_2017}.
The low-rank stochastic hierarchy (LRS) generalizes this construction
using a low-rank representation of bath correlations~\cite{zhang_lowrankhierarchicalframework_2026}.
Direct discretization of the functional derivative offers an
alternative~\cite{caizhu2026}.
Although these approaches allow individual trajectories to be computed,
recovering the reduced density still requires ensemble averaging, whose
slow sampling convergence can make high-accuracy calculations costly.

In this work, we use a low-rank decomposition of the bath correlation
function to obtain a finite-dimensional representation of the Gaussian
process. This reduces the infinite-dimensional functional derivative
to partial derivatives with respect to finitely many Gaussian
variables. By expanding the wavefunction in a polynomial chaos basis
and applying Galerkin projection, we obtain a deterministic coefficient
system with the same number of unknowns as a single HOPS trajectory.
The orthogonality of the basis functions then allows us to compute the
reduced density directly from these coefficients, without trajectory
sampling. This replaces an ensemble of $N_{\mathrm{traj}}$ trajectories
with one deterministic solve, reducing the computational cost by
approximately a factor of $N_{\mathrm{traj}}$.
Compared with HEOM, which is also deterministic and uses a density-matrix
formulation, the wavefunction formulation of DPCG requires substantially
fewer unknowns at the same truncation settings.

The main contributions of this work are as follows:
\begin{itemize}
 \item We develop a deterministic polynomial chaos Galerkin (DPCG)
 method for NMQSD. The method is easy to implement and provides a
 general numerical approach to efficient and accurate simulation of
 non-Markovian open quantum dynamics.
 The error estimate establishes the convergence of the method.
 \item DPCG replaces the $N_{\mathrm{traj}}$ trajectories used by HOPS
 with a single deterministic solve of the same size as one trajectory.
 This eliminates statistical sampling error and reduces the
 computational cost by approximately a factor of $N_{\mathrm{traj}}$.
 \item Benchmark comparisons show that DPCG achieves smaller errors
 at lower computational cost than the tested stochastic hierarchy
 methods. Further applications demonstrate that the same framework
 can handle challenging problems. 
\end{itemize}

Section~\ref{sec:nmqsd-background} introduces NMQSD and stochastic
hierarchies. Section~\ref{sec:chaos} develops DPCG and its error bound.
Sections~\ref{sec:numerics} and~\ref{sec:applications} present numerical
validation and applications, followed by conclusions in
Section~\ref{sec:conclusion}. Proofs and implementation details are
given in the appendices.

\section{Non-Markovian quantum state diffusion and stochastic hierarchies}
\label{sec:nmqsd-background}

\subsection{Non-Markovian quantum state diffusion}
\label{sec:nmqsd-equation}

The linear NMQSD equation reads~\cite{diosi_nonmarkovianquantumstate_1998a}
\begin{equation}
 \partial_t\psi(x,t;z)=
 \left[-iH+Lz_t^*-L^\dagger\int_0^t
 \alpha(t,s)\frac{\delta}{\delta z_s^*}\,\dd s\right]\psi(x,t;z).
 \label{eq:nmqsd-general}
\end{equation}
Here $\psi(\cdot,t;z)\in\mathcal H$ is the wavefunction,
$H$ is the system Hamiltonian, and $L$ is the system
coupling operator. $\alpha$ is the bath correlation function (BCF) and
$z_t$ is a centered complex Gaussian process with
$\E[z_tz_s]=0$ and $\E[z_tz_s^*]=\alpha(t,s)$.
The initial condition is $\psi(\cdot,0;z)=\psi_0$.
Writing $\rho_{\mathrm{tot}}(t)$ for the full system--bath state, the
reduced density matrix is
\begin{equation}
 \rho(t)=\Tr_{\mathrm{bath}}\rho_{\mathrm{tot}}(t)
 =\E\!\left[\psi(\cdot,t;z)\otimes
 \overline{\psi(\cdot,t;z)}\right].
\label{eq:reduced-density-general}
\end{equation}

The functional derivative in Eq.~\eqref{eq:nmqsd-general} acts on an
infinite-dimensional space of Gaussian process paths, making direct
numerical evaluation difficult. Stochastic hierarchy methods address this
difficulty by introducing auxiliary wavefunctions that satisfy coupled
evolution equations.

\subsection{Hierarchy of pure states}
\label{sec:hops-background}

HOPS uses an exponential approximation of the
BCF~\cite{suess_hierarchystochasticpure_2014}:
\begin{equation}
 \alpha(t,s)\approx\sum_{q=1}^{Q}g_q e^{-w_q(t-s)},\qquad t\geq s.
 \label{eq:hops-kernel-background}
\end{equation}
The auxiliary operators and wavefunctions are defined as
\begin{equation}
 \mathcal D_q(t)=\int_0^T g_q e^{-w_q(t-s)}
 \frac{\delta}{\delta z_s^*}\,\dd s,
 \qquad
 \psi^{\bm k}=\mathcal D_1^{k_1}\cdots\mathcal D_Q^{k_Q}\psi.
 \label{eq:hops-auxiliary-background}
\end{equation}
With $\bm e_q$ denoting the $q$th unit multi-index, the auxiliary
wavefunctions satisfy the linear hierarchy
\begin{equation}
 \partial_t\psi^{\bm k}
 =\left(-iH+Lz_t^*-\sum_{q=1}^{Q}k_qw_q\right)\psi^{\bm k}
 +L\sum_{q=1}^{Q}k_qg_q\psi^{\bm k-\bm e_q}
 -L^\dagger\sum_{q=1}^{Q}\psi^{\bm k+\bm e_q}.
 \label{eq:hops-background}
\end{equation}
The physical state is $\psi^{\bm0}=\psi$. In numerical calculations,
the hierarchy is truncated at total order $P$, retaining $|\bm k|\leq P$.
This formulation avoids direct evaluation of the functional derivative.
The number of auxiliary states grows combinatorially with $Q$ and $P$.
The reduced density is estimated by first solving individual trajectories
and then averaging their wavefunction outer products. Using finitely
many trajectories introduces statistical error.
The reliance on an exponential approximation of the BCF also limits
the practical applicability of HOPS.

\subsection{Low-rank stochastic hierarchy}
\label{sec:lrs-background}

LRS replaces the exponential
approximation with a low-rank decomposition
of the BCF~\cite{zhang_lowrankhierarchicalframework_2026}. The decomposition
and the corresponding auxiliary operators are
\begin{equation}
 \alpha(t,s)\approx\sum_{q=1}^{Q}b_q(t)b_q^*(s),\qquad
 \mathcal D_q=\int_0^T b_q^*(s)\frac{\delta}{\delta z_s^*}\,\dd s.
 \label{eq:lrs-kernel-background}
\end{equation}
The auxiliary wavefunctions $\psi^{\bm k}=\mathcal D_1^{k_1}\cdots
\mathcal D_Q^{k_Q}\psi$ satisfy
\begin{equation}
 \partial_t\psi^{\bm k}
 =(-iH+Lz_t^*)\psi^{\bm k}
 +L\sum_{q=1}^{Q}k_q b_q^*(t)\psi^{\bm k-\bm e_q}
 -L^\dagger\sum_{q=1}^{Q}b_q(t)\psi^{\bm k+\bm e_q}.
 \label{eq:lrs-background}
\end{equation}
LRS extends HOPS to more general BCFs by relaxing the requirement
of an exponential representation.
However, the number of auxiliary states still grows combinatorially
with $Q$ and $P$, and averaging over finitely many trajectories still
introduces statistical error.

\section{Deterministic polynomial chaos Galerkin formulation}
\label{sec:chaos}

\subsection{Low-rank decomposition}
\label{sec:low-rank-modes}

We first construct a finite-dimensional stochastic representation using a
low-rank approximation of the bath correlation function.
Fix $T>0$ and assume that $\alpha$ is a continuous Hermitian
positive-semidefinite kernel on $[0,T]^2$.
The eigenpairs $(\lambda_q,\varphi_q)_{q\geq1}$ satisfy
\begin{equation}
 \int_0^T\alpha(t,s)\varphi_q(s)\,\dd s
 =\lambda_q\varphi_q(t),
 \label{eq:covariance-eigenproblem}
\end{equation}
with
$\{\varphi_q\}$ orthonormal in $L^2(0,T;\mathbb C)$ and
$\lambda_1\geq\lambda_2\geq\cdots\geq0$. Set
\begin{equation}
 b_q(t)=\sqrt{\lambda_q}\,\varphi_q(t).
 \label{eq:kl-decomposition}
\end{equation}
The eigenvalues satisfy
$\sum_{q\geq1}\lambda_q=\int_0^T\alpha(t,t)\,\dd t<\infty$.

To represent $z_t$, let $\{\eta_q\}_{q\geq1}$ be independent standard
complex Gaussian variables on a probability space
$(\Omega,\mathcal F,\mathbb P)$, with $\eta_q\sim\CN(0,1)$. Then
$\E[\eta_q]=0$, $\E[\eta_q\eta_r]=0$, and
$\E[\eta_q\eta_r^*]=\delta_{qr}$.
The kernel expansion and the corresponding Karhunen--Lo\`eve~\cite{ghanem1991}
representation of the process are
\begin{equation}
 \alpha(t,s)=\sum_{q\geq1}b_q(t)b_q^*(s),\qquad
 z_t=\sum_{q\geq1}b_q(t)\eta_q.
 \label{eq:full-mode-expansions}
\end{equation}
Mercer's theorem~\cite{devito2013} gives uniform convergence of the kernel series on
$[0,T]^2$. The process $z_t$ satisfies $\E[z_tz_s]=0$ and
$\E[z_tz_s^*]=\alpha(t,s)$.
We use this representation of the original process throughout.

Retaining the first $Q$ modes in both expansions defines
\begin{equation}
 \widehat\alpha(t,s):=\sum_{q=1}^{Q}b_q(t)b_q^*(s),\qquad
 \widehat z_t:=\sum_{q=1}^{Q}b_q(t)\eta_q.
 \label{eq:finite-process}
\end{equation}
The truncated process $\widehat z_t$ satisfies:
\begin{equation}
 \E[\widehat z_t]=0,
 \qquad
 \E[\widehat z_t\widehat z_s]=0,
 \qquad
 \E[\widehat z_t\widehat z_s^*]=\widehat\alpha(t,s).
\label{eq:finite-process-covariance}
\end{equation}
Let $\psi(x,t;\widehat z)$ denote the solution of
Eq.~\eqref{eq:nmqsd-general} with $\alpha$ and $z$ replaced by
$\widehat\alpha$ and $\widehat z$, respectively.
The path dependence of $\psi(x,t;\widehat z)$ is now parameterized by
$\bm\eta^*=(\eta_1^*,\ldots,\eta_Q^*)$.
Since $\partial\widehat z_s^*/\partial\eta_q^*=b_q^*(s)$,
the chain rule gives
\begin{equation}
 \frac{\partial\psi(x,t;\widehat z)}{\partial\eta_q^*}
 =\int_0^t\frac{\delta\psi(x,t;\widehat z)}{\delta \widehat z_s^*}\cdot \frac{\partial \widehat z_s^*}{\partial\eta_q^*}\,\dd s
 = \int_0^t b_q^*(s)\frac{\delta\psi(x,t;\widehat z)}
 {\delta\widehat z_s^*}\,\dd s,
 \qquad q=1,\ldots,Q.
 \label{eq:chain-rule}
\end{equation}
Substituting the kernel factorization into the memory term therefore yields
\begin{equation}
 \int_0^t\widehat\alpha(t,s)
 \frac{\delta\psi(x,t;\widehat z)}{\delta\widehat z_s^*}\,\dd s
 =\int_0^t\left(\sum_{q=1}^{Q}b_q(t)b_q^*(s)\frac{\delta\psi(x,t;\widehat z)}{\delta\widehat z_s^*}\right)\dd s
 =\sum_{q=1}^{Q}b_q(t)
 \frac{\partial\psi(x,t;\widehat z)}{\partial\eta_q^*}.
 \label{eq:finite-memory-term}
\end{equation}
The resulting equation is
\begin{equation}
 \partial_t\psi(x,t;\widehat z)
 =-iH\psi(x,t;\widehat z)
 +L\sum_{q=1}^{Q}b_q^*(t)\eta_q^*\psi(x,t;\widehat z)
 -L^\dagger\sum_{q=1}^{Q}b_q(t)
 \frac{\partial\psi(x,t;\widehat z)}{\partial\eta_q^*}.
 \label{eq:finite-nmqsd}
\end{equation}
Equation~\eqref{eq:finite-nmqsd} expresses the memory term through partial
derivatives with respect to the Gaussian variables $\eta_1^*,\ldots,\eta_Q^*$.

\subsection{Polynomial chaos Galerkin projection}
\label{sec:general-projection}

Following the stochastic Galerkin construction~\cite{ghanem1991,xiu2002,young_simulationstochasticquantum_2013},
we consider an expansion of the wavefunction in a given set of basis
functions $\{\mathcal P_{\bm k}(\bm\eta^*)\}_{\bm k\in\mathcal K_{Q,P}}$. Here we choose normalized monomials in $\bm\eta^*$ that
are orthonormal with respect to the Gaussian measure. The resulting
polynomial chaos expansion is
\begin{equation}
 \psi_m(x,t;\widehat z):=\sum_{\bm k\in\mathcal K_{Q,P}}
 u_{\bm k}(x,t)\mathcal P_{\bm k}(\bm\eta^*),
 \qquad u_{\bm k}(\cdot,t)\in\mathcal H.
 \label{eq:chaos-expansion}
\end{equation}
For a total-degree cutoff $P\in\mathbb N_0$, the retained multi-index set
and the number $m$ of basis functions are
\begin{equation}
 \mathcal K_{Q,P}=\{\bm k\in\mathbb N_0^Q:|\bm k|\le P\},
 \qquad
 m:=|\mathcal K_{Q,P}|=\binom{Q+P}{P},
 \label{eq:total-order}
\end{equation}
where $\mathbb N_0=\{0,1,2,\ldots\}$ and $|\bm k|=\sum_{q=1}^{Q}k_q$.

The polynomial chaos~\cite{bargmann1961} basis functions are
\begin{equation}
 \mathcal P_{\bm k}(\bm\eta^*)=
 \prod_{q=1}^{Q}\frac{(\eta_q^*)^{k_q}}{\sqrt{k_q!}},
 \qquad
 \E[\mathcal P_{\bm k}\mathcal P_{\bm\ell}^*]
 =\delta_{\bm k\bm\ell},
 \label{eq:bargmann-basis}
\end{equation}
where $\bm k,\bm\ell\in\mathbb N_0^Q$. Their orthonormality follows from independence
and the complex Gaussian moment identity
$\E[\eta^k(\eta^*)^\ell]=\delta_{k\ell}k!$ for $\eta\sim\CN(0,1)$.
The basis functions satisfy
\begin{equation}
 \eta_q^*\mathcal P_{\bm k}
 =\sqrt{k_q+1}\,\mathcal P_{\bm k+\bm e_q},
 \qquad
 \frac{\partial\mathcal P_{\bm k}}{\partial\eta_q^*}
 =\sqrt{k_q}\,\mathcal P_{\bm k-\bm e_q},
 \label{eq:ladder-rules}
\end{equation}
where $\bm e_q$ is the $q$th coordinate vector, and the derivative is zero
when $k_q=0$.
Substituting Eq.~\eqref{eq:chaos-expansion} into
Eq.~\eqref{eq:finite-nmqsd}, using Eq.~\eqref{eq:ladder-rules}, and
projecting onto the retained polynomial space gives
\begin{equation}
 \begin{aligned}
 \sum_{\bm\ell\in\mathcal K_{Q,P}}
 (\partial_tu_{\bm\ell})\mathcal P_{\bm\ell}
 ={}&-iH\sum_{\bm\ell\in\mathcal K_{Q,P}}
 u_{\bm\ell}\mathcal P_{\bm\ell}
 +L\sum_{q=1}^{Q}b_q^*(t)
 \sum_{\substack{\bm\ell\in\mathcal K_{Q,P}\\|\bm\ell|<P}}
 \sqrt{\ell_q+1}\,u_{\bm\ell}\mathcal P_{\bm\ell+\bm e_q}\\
 &-L^\dagger\sum_{q=1}^{Q}b_q(t)
 \sum_{\substack{\bm\ell\in\mathcal K_{Q,P}\\\ell_q>0}}
 \sqrt{\ell_q}\,u_{\bm\ell}\mathcal P_{\bm\ell-\bm e_q}.
 \end{aligned}
 \label{eq:chaos-substitution}
\end{equation}
Multiply both sides of Eq.~\eqref{eq:chaos-substitution} by
$\mathcal P_{\bm k}^*$ and take expectation. By orthonormality, only the
terms with $\bm\ell=\bm k$ remain in the time derivative and Hamiltonian
sums, while only those with $\bm\ell=\bm k-\bm e_q$ and
$\bm\ell=\bm k+\bm e_q$ remain in the multiplication and differentiation
sums, respectively. Thus, for
$\bm k\in\mathcal K_{Q,P}$, the NMQSD equation is transformed into the following deterministic differential equation:
\begin{equation}
 \partial_tu_{\bm k}(x,t)=-iHu_{\bm k}(x,t)
 +L\sum_{q=1}^{Q}b_q^*(t)\sqrt{k_q}\,u_{\bm k-\bm e_q}(x,t)
 -L^\dagger\sum_{q=1}^{Q}b_q(t)\sqrt{k_q+1}\,
 u_{\bm k+\bm e_q}(x,t).
 \label{eq:chaos-hierarchy}
\end{equation}
In the truncated coefficient system, coefficients with indices outside
$\mathcal K_{Q,P}$ are set to zero.
The initial conditions are
$u_{\bm0}(x,0)=\psi_0(x)$ and
$u_{\bm k}(x,0)=0$ for $\bm k\ne\bm0$.
Orthonormality gives the coordinate density and the reduced density matrix:
\begin{equation}
 \rho_m(x,t)=\sum_{\bm k\in\mathcal K_{Q,P}}|u_{\bm k}(x,t)|^2,
 \qquad
 \rho_m(t)=\sum_{\bm k\in\mathcal K_{Q,P}}
 |u_{\bm k}(\cdot,t)\rangle\langle u_{\bm k}(\cdot,t)|.
 \label{eq:chaos-density}
\end{equation}
The reduced density matrix $\rho_m(t)$ is positive semidefinite,
and its trace is conserved:
\begin{equation}
 \frac{\dd}{\dd t}\Tr\rho_m(t)=0.
 \label{eq:mass-conservation}
\end{equation}
The proof is given in Appendix~\ref{app:trace-conservation}.

\subsection{Numerical discretization}
\label{sec:discretization-overview}
\label{sec:implementation}

The coefficients in Eq.~\eqref{eq:chaos-hierarchy} take values in the
system Hilbert space $\mathcal H$. We distinguish two settings:
$\mathcal H=\mathbb C^{N_{\mathrm s}}$ for a finite-dimensional quantum
system, and
$\mathcal H=L^2(\mathbb R^d;\mathbb C)$ for a system with $d$ continuous
coordinates. The first setting requires only time discretization.
In the second setting, we first discretize the spatial dependence of
each coefficient wavefunction, then advance the resulting finite arrays
in time.

\subsubsection{Finite-dimensional state space}
\label{sec:finite-system-implementation}
\label{sec:time-discretization}

Let $\mathcal H=\mathbb C^{N_{\mathrm s}}$, where $N_{\mathrm s}$ is
the dimension of the original system state space. Then $H$ and $L$
are $N_{\mathrm s}\times N_{\mathrm s}$ matrices and
$u_{\bm k}(t)\in\mathbb C^{N_{\mathrm s}}$. Collect these vectors,
$\bm k\in\mathcal K_{Q,P}$, in an array
$\bm U(t)\in\mathbb C^{N_{\mathrm s}\times m}$.
With $F$ denoting the right-hand side of
Eq.~\eqref{eq:chaos-hierarchy}, the system is
\begin{equation}
 \frac{\dd\bm U}{\dd t}=F(t,\bm U).
 \label{eq:coefficient-ode}
\end{equation}
Set $t_n=n\Delta t$ and let $\bm U^n$ approximate $\bm U(t_n)$.
The classical fourth-order Runge--Kutta scheme (RK4)~\cite{butcher_numericalmethodsordinary_2016} uses
\begin{equation}
 \begin{aligned}
 \bm F_1={}&F(t_n,\bm U^n),\\
 \bm F_2={}&F(t_n+\tfrac12\Delta t,
                 \bm U^n+\tfrac12\Delta t\,\bm F_1),\\
 \bm F_3={}&F(t_n+\tfrac12\Delta t,
                 \bm U^n+\tfrac12\Delta t\,\bm F_2),\\
 \bm F_4={}&F(t_n+\Delta t,\bm U^n+\Delta t\,\bm F_3),\\
 \bm U^{n+1}={}&\bm U^n+
 \frac{\Delta t}{6}(\bm F_1+2\bm F_2+2\bm F_3+\bm F_4).
 \end{aligned}
 \label{eq:rk4-scheme}
\end{equation}
\subsubsection{Continuous-coordinate state space}
\label{sec:spatial-implementation}

Let $\mathcal H=L^2(\mathbb R;\mathbb C)$, so that
$u_{\bm k}(\cdot,t)\in L^2(\mathbb R;\mathbb C)$.
For spatial discretization, choose a $K$-dimensional subspace
\begin{equation}
 \mathcal V_K=\operatorname{span}\{\chi_0,\ldots,\chi_{K-1}\}
 \subset L^2(\mathbb R;\mathbb C),
 \label{eq:spatial-subspace}
\end{equation}
with an orthonormal basis in the common domains of $H$, $L$, and
$L^\dagger$. Approximate each coefficient wavefunction in
$\mathcal V_K$ by
\begin{equation}
 u_{\bm k}(x,t)\approx
 \sum_{j=0}^{K-1}c_{\bm k,j}(t)\chi_j(x),
 \qquad
 \bm c_{\bm k}(t)=
 (c_{\bm k,0}(t),\ldots,c_{\bm k,K-1}(t))^{\mathsf T}.
 \label{eq:spatial-expansion}
\end{equation}
For example, the normalized Hermite functions are~\cite{shen2011}
\begin{equation}
 \chi_n(x)=\frac{\mathrm H_n(x)e^{-x^2/2}}
 {\pi^{1/4}\sqrt{2^n n!}},
 \qquad n\in\mathbb N_0,
 \label{eq:spatial-hermite-functions}
\end{equation}
where $\mathrm H_n(x)=(-1)^ne^{x^2}\frac{\dd^n}{\dd x^n}e^{-x^2}$
is the Hermite polynomial. These functions form an
orthonormal basis of $L^2(\mathbb R;\mathbb C)$.

Define the matrices
\begin{equation}
 (H_K)_{ij}=\langle\chi_i,H\chi_j\rangle,\qquad
 (L_K)_{ij}=\langle\chi_i,L\chi_j\rangle.
 \label{eq:projected-operators}
\end{equation}
Substituting Eq.~\eqref{eq:spatial-expansion} into
Eq.~\eqref{eq:chaos-hierarchy} and taking inner products with
$\chi_i$ gives
\begin{equation}
 \partial_t\bm c_{\bm k}
 =-iH_K\bm c_{\bm k}
 +L_K\sum_{q=1}^{Q}b_q^*(t)\sqrt{k_q}\,\bm c_{\bm k-\bm e_q}
 -L_K^\dagger\sum_{q=1}^{Q}b_q(t)\sqrt{k_q+1}\,
 \bm c_{\bm k+\bm e_q}.
 \label{eq:spatial-coefficient-equations}
\end{equation}
Coefficients outside $\mathcal K_{Q,P}$ are set to zero.
The initial coefficients are, for $j=0,\ldots,K-1$,
\begin{equation}
 c_{\bm0,j}(0)=\langle\chi_j,\psi_0\rangle,
 \qquad c_{\bm k,j}(0)=0\quad(\bm k\ne\bm0).
 \label{eq:spatial-initial-coefficients}
\end{equation}
If unit initial trace is imposed, the nonzero vector $\bm c_{\bm0}(0)$
is divided by its Euclidean norm.

We advance Eq.~\eqref{eq:spatial-coefficient-equations} using
Eq.~\eqref{eq:rk4-scheme}.
With $c_{\bm k,j}^n\approx c_{\bm k,j}(t_n)$, the coordinate density
is reconstructed as
\begin{equation}
 \rho_m(x,t_n)\approx
 \sum_{\bm k\in\mathcal K_{Q,P}}\left|
 \sum_{j=0}^{K-1}c_{\bm k,j}^n\chi_j(x)\right|^2.
 \label{eq:spatial-density-reconstruction}
\end{equation}
For higher-dimensional coordinates, the same construction uses tensor
products of one-dimensional spatial basis functions.

\subsection{Convergence}
\label{sec:error-bound}

Let $\psi_{t_n}$ be the exact state at $t_n=n\Delta t\in[0,T]$,
and let $\psi^{Q,P,n}$ be the numerical solution reconstructed from the
DPCG coefficients. We measure the state error in the norm
$\norm{v}_{\mathcal X}=\sqrt{\E\norm{v}_{\mathcal H}^2}$.
The following theorem gives the error bound for the DPCG approximation.

\begin{theorem}
\label{thm:main-error}
Let $H$ and $L$ be bounded operators, with $H$ self-adjoint.
Under Assumptions~\ref{ass:convergence} and~\ref{ass:numerical-accuracy},
at the time nodes $t_n\in[0,T]$, the exact state $\psi_{t_n}$ and the
numerical state $\psi^{Q,P,n}$ satisfy
\begin{equation}
 \max_n\norm{\psi_{t_n}-\psi^{Q,P,n}}_{\mathcal X}
 \leq C_1\left(\sqrt{\tau_Q}
 +\frac{C_2^{P+1}}{\sqrt{(P+1)!}}\right)
 +\delta_{\mathrm{num}},
 \label{eq:main-error}
\end{equation}
where $Q\geq1$ is the low-rank truncation order, $P\geq0$ is the
polynomial total-degree cutoff, and $\tau_Q=\sum_{q>Q}\lambda_q$
is the discarded covariance eigenvalue sum.
The quantity $\delta_{\mathrm{num}}$ is the numerical coefficient error.
The constants $C_1$ and $C_2$ depend on $T$, $L$, and $\alpha$.
\end{theorem}

The bound relates the low-rank truncation error to the discarded
covariance eigenvalues and gives a rapidly decaying bound for the
polynomial truncation error as $P\to\infty$. Together with the numerical
coefficient error, these terms control the state error and hence the
error in the reduced density matrix. The proof is given in
Appendix~\ref{app:convergence}.

\section{Accuracy and computational cost}
\label{sec:numerics}

We assess the accuracy and computational cost of DPCG through convergence
tests and comparisons with reference solutions and existing numerical
methods.

Unless stated otherwise, the coefficient equations are solved by RK4 in
the interaction picture. To remove the explicit Hamiltonian term from the
evolution equation, we set
\begin{equation}
 \hat{\psi}=e^{iHt}\psi,
 \qquad \hat{L}(t)=e^{iHt}Le^{-iHt}.
 \label{eq:interaction-picture}
\end{equation}
Then Eq.~\eqref{eq:nmqsd-general} becomes
\begin{equation}
 \partial_t\hat{\psi}(t;z)
 =\hat{L}(t)z_t^*\hat{\psi}(t;z)
 -\hat{L}^\dagger(t)\int_0^t\alpha(t,s)
 \frac{\delta\hat{\psi}(t;z)}{\delta z_s^*}\,\dd s.
 \label{eq:interaction-picture-nmqsd}
\end{equation}
The initial state and bath correlation are unchanged. The low-rank and
chaos constructions in Section~\ref{sec:chaos} apply with $\hat{L}(t)$ in
place of $L$ and no Hamiltonian term.
The resulting chaos coefficients $\hat u_{\bm k}$ are related to the original
coefficients by $u_{\bm k}(\cdot,t)=e^{-iHt}\hat u_{\bm k}(\cdot,t)$.

For the low-rank decomposition, we discretize the covariance kernel
on uniform time nodes $s_i\in[0,T]$ with spacing $\Delta s$ and
positive weights $w_i$. The retained eigenpairs
$(\lambda_q,v_q)$, normalized by $v_q^\dagger v_q=1$, give
\begin{equation}
 A_{ij}=\sqrt{w_i}\alpha(s_i,s_j)\sqrt{w_j},
 \qquad Av_q=\lambda_qv_q,
 \qquad
 \widetilde b_q(s_i)=\frac{\sqrt{\lambda_q}\,v_{q,i}}{\sqrt{w_i}}.
 \label{eq:numerical-kernel-factors}
\end{equation}
The spin calculation uses $w_i=1$ and the Caldeira--Leggett calculations
use trapezoidal weights.
The decomposition is computed before time integration, and its grid
can be chosen independently of the time-stepping grid.

OpenAI Codex (GPT-6 Astra) was used to assist the authors in implementing and debugging the numerical codes. The authors reviewed the code and verified the numerical results.

Table~\ref{tab:numerical-notation} summarizes the main notation used below.
\begin{table}[htbp]
 \centering
 \caption{Main notation for the numerical experiments.}
 \label{tab:numerical-notation}
 \small
 \setlength{\tabcolsep}{6pt}
 \begin{tabular}{@{}cp{0.32\linewidth}@{\hspace{1.5em}}cp{0.34\linewidth}@{}}
  \toprule
  Symbol & Meaning & Symbol & Meaning\\
  \midrule
  $H$ & System Hamiltonian
      & $\varepsilon$ & Dimensionless Planck constant\\
  $L$ & System coupling operator
      & $\xi$ & Coupling strength\\
  $\alpha$ & Bath correlation function
      & $\beta$ & Inverse temperature\\
  $\psi$ & Stochastic wavefunction
      & $Q$ & Retained low-rank modes per bath\\
  $\psi_0$ & Initial wavefunction
      & $P$ & Maximum total chaos degree\\
  $u_{\bm k}$ & Chaos expansion coefficients
      & $m$ & Chaos basis size\\
  $\rho(t)$ & Reduced density matrix
      & $K$ & Total spatial basis size\\
  $\rho(x,t)$ & Coordinate density
      & $\Delta t$ & Time step\\
  \bottomrule
 \end{tabular}
\end{table}

\subsection{Spin--boson model}
\label{sec:spin-numerics}

We first compare the accuracy and computational cost of DPCG, HOPS,
and LRS for the spin--boson model, which describes a two-level system
coupled to a bosonic bath~\cite{leggett_dynamicsdissipativetwostate_1987,hartmann_exactopenquantum_2017,strathearn_efficientnonmarkovianquantum_2018a}.

\begin{example}[Spin--boson dynamics]
\label{ex:spin-boson}

The system Hamiltonian and coupling operator are
\begin{equation}
 H=\begin{pmatrix}\epsilon&\Delta\\\Delta&-\epsilon\end{pmatrix},
 \qquad L=\sigma_z=\begin{pmatrix}1&0\\0&-1\end{pmatrix},
 \label{eq:spin-realization}
\end{equation}
where $\epsilon$ is the spin bias and $\Delta$ is the tunneling parameter.
For a thermal bosonic bath with an Ohmic spectral density and cutoff
frequency $\omega_c$, the correlation function is
\begin{equation}
 \begin{aligned}
 \alpha(t,s)&=\int_0^\infty J(\omega)
 \left[\coth\!\left(\frac{\beta\omega}{2}\right)\cos\omega(t-s)
       -i\sin\omega(t-s)\right]\dd\omega,\\
 J(\omega)&=\frac{\xi}{2}\omega e^{-\omega/\omega_c}.
 \end{aligned}
 \label{eq:spin-bcf}
\end{equation}
We take $\psi_0=(1,0)^{\mathsf T}$ and
\begin{equation}
 \epsilon=0,\qquad \Delta=1,\qquad \xi\in\{0.2,0.8\},\qquad
 \omega_c=2.5,\qquad \beta=5,\qquad T=5.
 \label{eq:spin-parameters}
\end{equation}
All methods use the same finite bath, obtained by discretizing
Eq.~\eqref{eq:spin-bcf} with $N_{\mathrm b}=400$ frequencies up to
$\omega_{\max}=4\omega_c$:
\begin{equation}
 \begin{aligned}
 \omega_j&=-\omega_c\log\!\left[
 1-\frac{j}{N_{\mathrm b}}(1-e^{-\omega_{\max}/\omega_c})\right],\qquad
 c_j=\omega_j\sqrt{\frac{\xi\omega_c}{N_{\mathrm b}}
 (1-e^{-\omega_{\max}/\omega_c})},\\[-3pt]
 \alpha(t,s)&=\sum_{j=1}^{N_{\mathrm b}}\frac{c_j^2}{2\omega_j}
 \left[\coth\!\left(\frac{\beta\omega_j}{2}\right)\cos\omega_j(t-s)
       -i\sin\omega_j(t-s)\right].
 \end{aligned}
 \label{eq:spin-finite-bath}
\end{equation}

DPCG, HOPS, and LRS use $P=6$, $Q=12$, and $\Delta t=0.025$, with
$1000$ trajectories for each stochastic method.
The reference solution is computed using the process tensor
time-evolving matrix product operator (PT-TEMPO)
method~\cite{strathearn_efficientnonmarkovianquantum_2018a,jorgensen2019}.
We use a time step of $0.05$, a memory time of $5$, and relative
compression tolerances of $10^{-9}$ and $10^{-8}$ for $\xi=0.2$ and
$\xi=0.8$, respectively.
\end{example}

Figure~\ref{fig:spin-population} and Table~\ref{tab:spin-comparison}
compare the accuracy and computational cost for the spin polarization
$\langle\sigma_z\rangle_t$.
For $\xi=0.2$, all three methods reproduce the damped oscillations;
at the stronger bath coupling $\xi=0.8$, the oscillations are strongly
suppressed, and HOPS exhibits the largest error at late times.
DPCG remains close to the reference in both cases, with smaller errors
and shorter computation times than HOPS and LRS.
For these stochastic methods, suppressing sampling error through a
larger ensemble entails additional trajectory propagations.
DPCG avoids this tradeoff by recovering the averaged observable
directly from its deterministic coefficients, without sampling error.
It also reproduces the PT-TEMPO dynamics at a substantially lower
computational cost.

\begin{figure}[htbp]
 \centering
 \includegraphics[width=\linewidth]{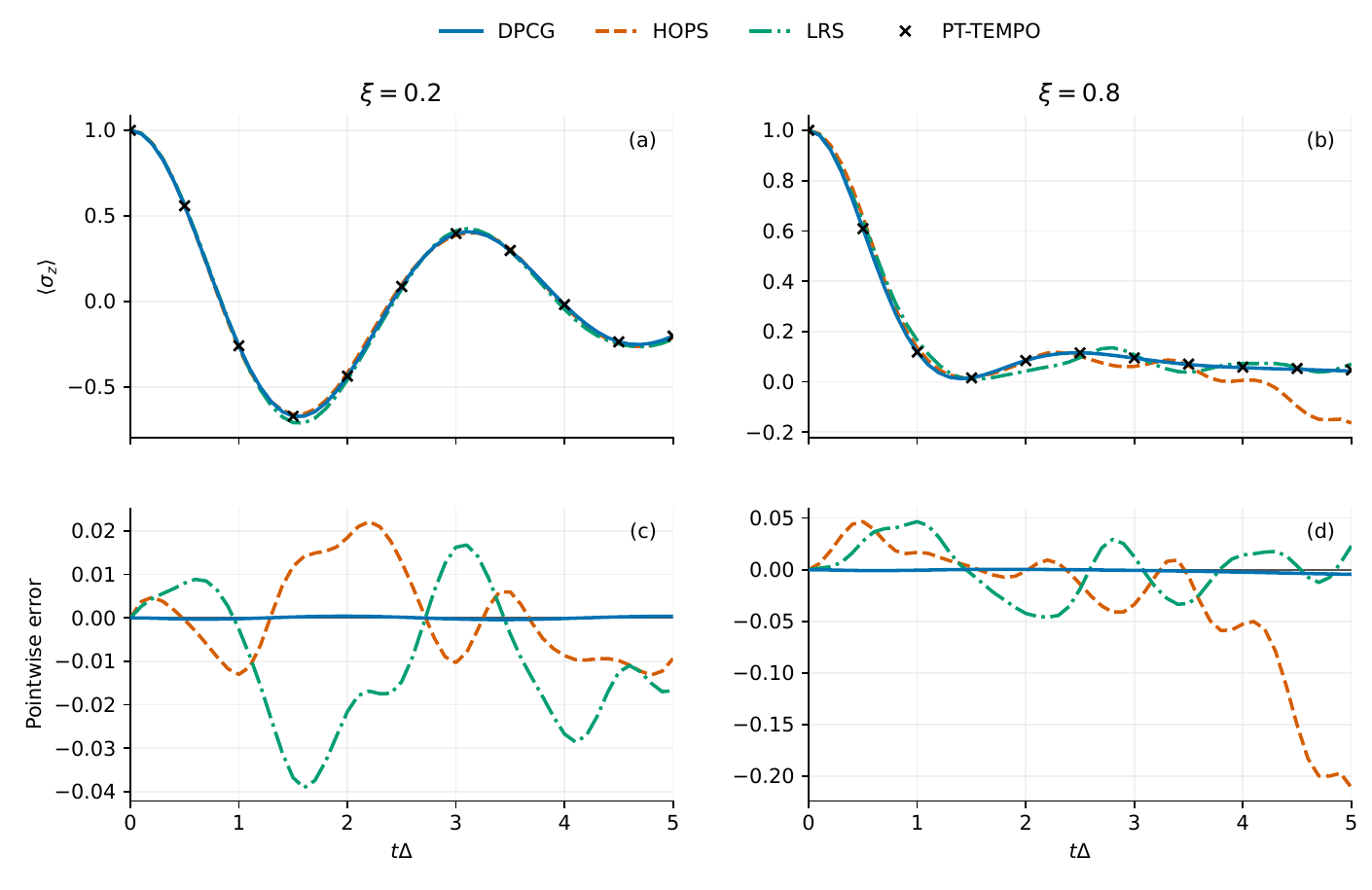}
 \caption{Example~\ref{ex:spin-boson}. Spin polarization (top) and
 pointwise error relative to PT-TEMPO (bottom) for $\xi=0.2$ (left) and $\xi=0.8$
 (right). Crosses mark the PT-TEMPO reference.}
 \label{fig:spin-population}
\end{figure}

\begin{table}[htbp]
 \centering
 \caption{Example~\ref{ex:spin-boson}. Maximum absolute error in
 $\langle\sigma_z\rangle$ relative to PT-TEMPO and computation time. HOPS and LRS use $1000$ trajectories.}
 \label{tab:spin-comparison}
 \setlength{\tabcolsep}{10pt}
 \begin{tabular}{lcrcr}
  \toprule
  & \multicolumn{2}{c}{$\xi=0.2$}
      & \multicolumn{2}{c}{$\xi=0.8$}\\
  \cmidrule(lr){2-3}\cmidrule(lr){4-5}
  Method & Error & Time (s) & Error & Time (s)\\
  \midrule
  \inputtablerows{tables/spin_boson_comparison_rows.tex}
  \bottomrule
 \end{tabular}
\end{table}

\subsection{Caldeira--Leggett model}
\label{sec:numerical-setup}

The Caldeira--Leggett model describes a particle in a general potential,
linearly coupled to a bath of harmonic oscillators~\cite{caldeira1983,wang2025,thorwart2000,zhan2026}.
With dimensionless Planck constant
$\varepsilon$, the total Hamiltonian is
\begin{equation}
 \begin{aligned}
 H_{\mathrm{tot}}={}&H_{\mathrm{sys}}\otimes I_{\mathrm{bath}}
 +I_{\mathrm{sys}}\otimes H_{\mathrm{bath}}
 +W_{\mathrm{s}}\otimes W_{\mathrm{b}},\\
 H_{\mathrm{sys}}={}&-\frac{\varepsilon^2}{2}\partial_x^2+V(x),
 & (W_{\mathrm{s}}f)(x)={}&xf(x),\\
 H_{\mathrm{bath}}={}&\sum_{j=1}^{N_{\mathrm{b}}}
 \left(-\frac{\varepsilon^2}{2}\partial_{r_j}^2
 +\frac{\omega_j^2r_j^2}{2}\right),
 & W_{\mathrm{b}}={}&-\sum_{j=1}^{N_{\mathrm{b}}}c_jr_j.
 \end{aligned}
 \label{eq:cl-model}
\end{equation}
Here $r_j$ are the bath coordinates. The potential includes the counterterm,
\begin{equation}
 \kappa=\sum_{j=1}^{N_{\mathrm{b}}}\frac{c_j^2}{\omega_j^2},
 \qquad
 V(x)=V_0(x)+\frac{\kappa}{2}x^2,
 \label{eq:effective-potential}
\end{equation}
where $V_0$ is the bare system potential.

The corresponding Schr\"odinger equation reads
$i\varepsilon\partial_t\Psi_t=H_{\mathrm{tot}}\Psi_t$.
For a thermal bath, the reduced system dynamics can be equivalently
described by the NMQSD equation~\eqref{eq:nmqsd-general}, with
\begin{equation}
 \begin{aligned}
 H&=H_{\mathrm{sys}}/\varepsilon,\qquad (Lf)(x)=xf(x),\\
 \alpha(t,s)&=B_\varepsilon(t,s)
 =\frac12\sum_{j=1}^{N_{\mathrm b}}
 \frac{c_j^2}{\varepsilon\omega_j}
 \left[\coth\!\left(\frac{\beta\varepsilon\omega_j}{2}\right)
 \cos\omega_j(t-s)-i\sin\omega_j(t-s)\right].
 \end{aligned}
 \label{eq:thermal-bcf}
\end{equation}
We use the bath frequencies $\omega_j$ from Eq.~\eqref{eq:spin-finite-bath}
and rescale the corresponding coupling coefficients as
$c_j\mapsto\varepsilon c_j$.
Unless stated otherwise,
\begin{equation}
 N_{\mathrm{b}}=400,\quad \omega_{\max}=10,\quad
 \omega_c=2.5,\quad \beta=5,\quad T=5.
 \label{eq:numerical-parameters}
\end{equation}

\begin{example}[Harmonic dynamics]
\label{ex:cl-harmonic}

We specialize the bare system potential to $V_0(x)=x^2/2$.
For spatial discretization, we use scaled Hermite functions
$a\chi_n(a^2x)$, where $a=(1+\kappa)^{1/8}\varepsilon^{-1/4}$.
Consider the following two initial states:
\begin{itemize}
\item A single wave packet,
\begin{equation}
 \psi_0(x)=(4\pi\varepsilon)^{-1/4}
 \exp\!\left(-\frac{x^2}{8\varepsilon}+\frac{ix}{4\varepsilon}\right).
 \label{eq:single-packet-state}
\end{equation}
Test the accuracy of the computed density at $\varepsilon=1/32$
and $\xi=1$, using $Q=16$, $P=6$, $K=48$, and $\Delta t=0.025$.

\item A superposition of two wave packets,
\begin{equation}
 \psi_0(x)=C\left[
 \exp\!\left(-\frac{(x-1/2)^2}{4\varepsilon}\right)
 +\frac45\exp\!\left(-\frac{(x+1/2)^2}{4\varepsilon}\right)
 \right],
 \quad
 C=\frac{5(2\pi\varepsilon)^{-1/4}}
 {\sqrt{41+40e^{-\frac{1}{8\varepsilon}}}}.
\label{eq:two-peak-state}
\end{equation}
We examine the effects of $\varepsilon$ and $\xi$ on the density evolution
from a fixed initial state.
Here, $\varepsilon$ affects both the system and bath, while $\xi$
controls only the bath coupling.

In both parameter studies, we fix the initial state by setting
$\varepsilon=\varepsilon_0=1/64$ in Eq.~\eqref{eq:two-peak-state}.

For the $\varepsilon$-dependence, we fix $\xi=1$ and use
$Q=12$, $K=48$, and $\Delta t=0.025$, with $P=6$ for
$\varepsilon=1/64,1/40,1/24,1/16$ and $P=8$ for
$\varepsilon=1/8,1/4,1/2$.

For the $\xi$-dependence, we fix $\varepsilon=1/64$
and take $\xi=0,3.2,6.4,12.8$.
All cases use $Q=10$, $P=12$, $K=48$, and $\Delta t=0.025$.
\end{itemize}
The reference density for this harmonic case is~\cite{fleming2011}
\begin{equation}
 \begin{aligned}
 \rho^*(x,t)=\frac{1}{2\pi}\int_{\mathbb R}\int_{\mathbb R}
 &e^{ik(F_{00}(t)y-x)-\sigma_{\mathrm b}^2(t)k^2/2}\\
 &\times\overline{\psi_0\!\left(y-\frac{\varepsilon G_{00}(t)k}{2}\right)}
 \psi_0\!\left(y+\frac{\varepsilon G_{00}(t)k}{2}\right)\,\dd y\,\dd k,
 \end{aligned}
 \label{eq:exact-harmonic-density}
\end{equation}
where
\begin{equation}
 \begin{aligned}
 \bm F(t)&=\cos\!\left(t\bm K_{\mathrm{CL}}^{1/2}\right),
 \qquad
 \bm G(t)=\bm K_{\mathrm{CL}}^{-1/2}
 \sin\!\left(t\bm K_{\mathrm{CL}}^{1/2}\right),\\
 \bm K_{\mathrm{CL}}&=
 \begin{pmatrix}
  1+\kappa&-\bm c^{\mathsf T}\\
  -\bm c&\operatorname{diag}(\omega_1^2,\ldots,\omega_{N_{\mathrm b}}^2)
 \end{pmatrix},
 \qquad \bm c=(c_1,\ldots,c_{N_{\mathrm b}})^{\mathsf T},\\
 \sigma_{\mathrm b}^2(t)&=\frac{\varepsilon}{2}
 \sum_{j=1}^{N_{\mathrm b}}
 \coth\!\left(\frac{\beta\varepsilon\omega_j}{2}\right)
 \left(\frac{F_{0j}(t)^2}{\omega_j}+\omega_jG_{0j}(t)^2\right).
 \end{aligned}
 \label{eq:exact-harmonic-coefficients}
\end{equation}
\end{example}

The single packet oscillates about the potential minimum at $x=0$,
as shown in Figure~\ref{fig:exact-validation}.
The numerical solution agrees closely with the reference solution at all
displayed times.

\begin{figure}[htbp]
 \centering
 \includegraphics[width=0.98\linewidth]
 {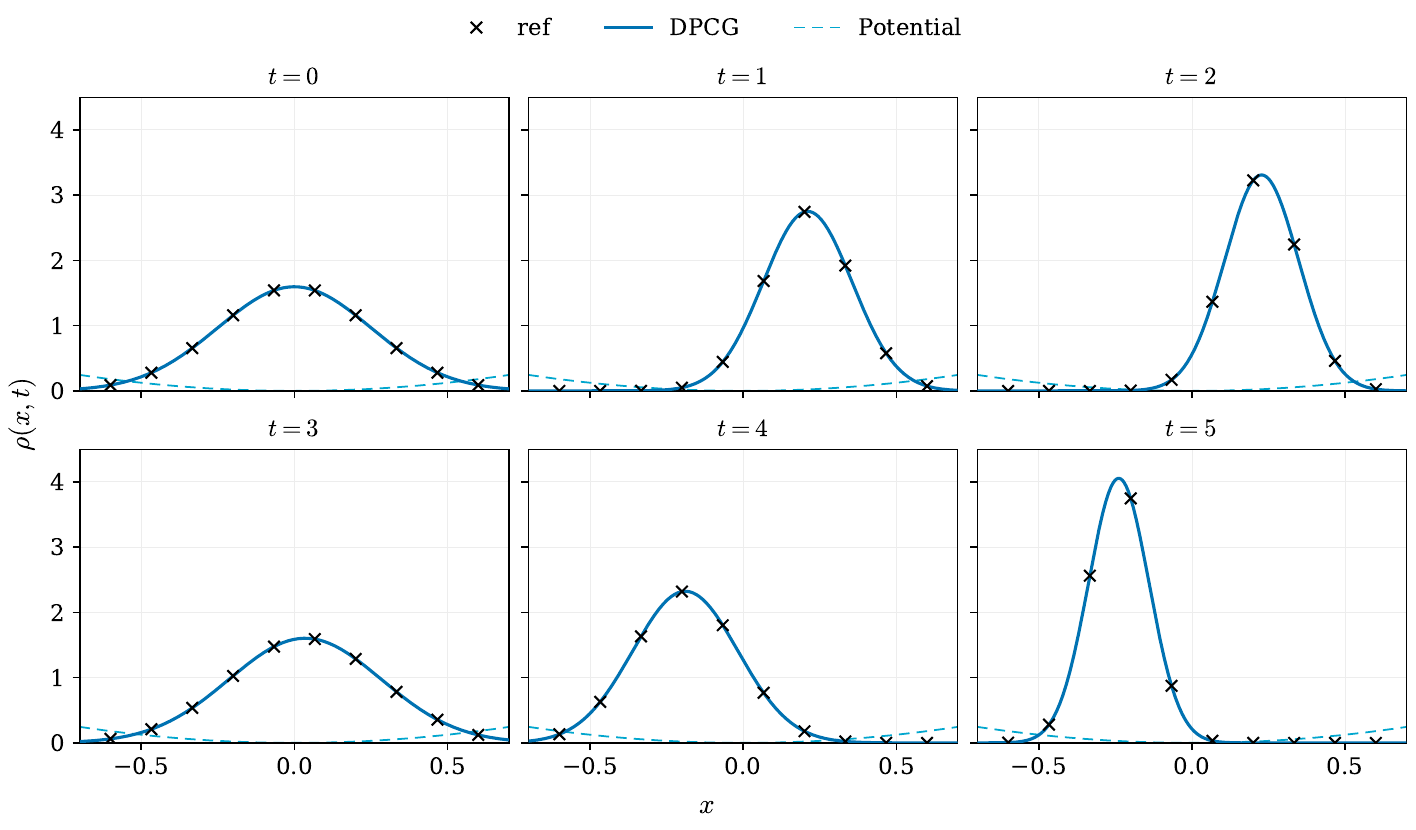}
 \caption{Example~\ref{ex:cl-harmonic}. Single-packet harmonic dynamics at
 $\varepsilon=1/32$ and $\xi=1$.
 Blue solid: DPCG. Black crosses: reference solution.
 Cyan dashed: the bare potential $V_0(x)=x^2/2$.}
 \label{fig:exact-validation}
\end{figure}

Figure~\ref{fig:epsilon-dependence} shows how the density at $T=5$ changes
with $\varepsilon$ for the fixed two-packet initial state.
Several peaks remain visible at small $\varepsilon$.
As $\varepsilon$ increases, the density broadens
and the central peak becomes lower, while the smaller peaks gradually
disappear.
At $\varepsilon=1/2$, the density has a smooth, approximately Gaussian
profile.
At $\varepsilon=1/64$, increasing the bath coupling $\xi$ also lowers
the peaks and suppresses the oscillations between them,
but the two main peaks remain distinct
(Figure~\ref{fig:xi-dependence-small}).

\begin{figure}[htbp]
 \centering
 \includegraphics[width=0.98\linewidth]
 {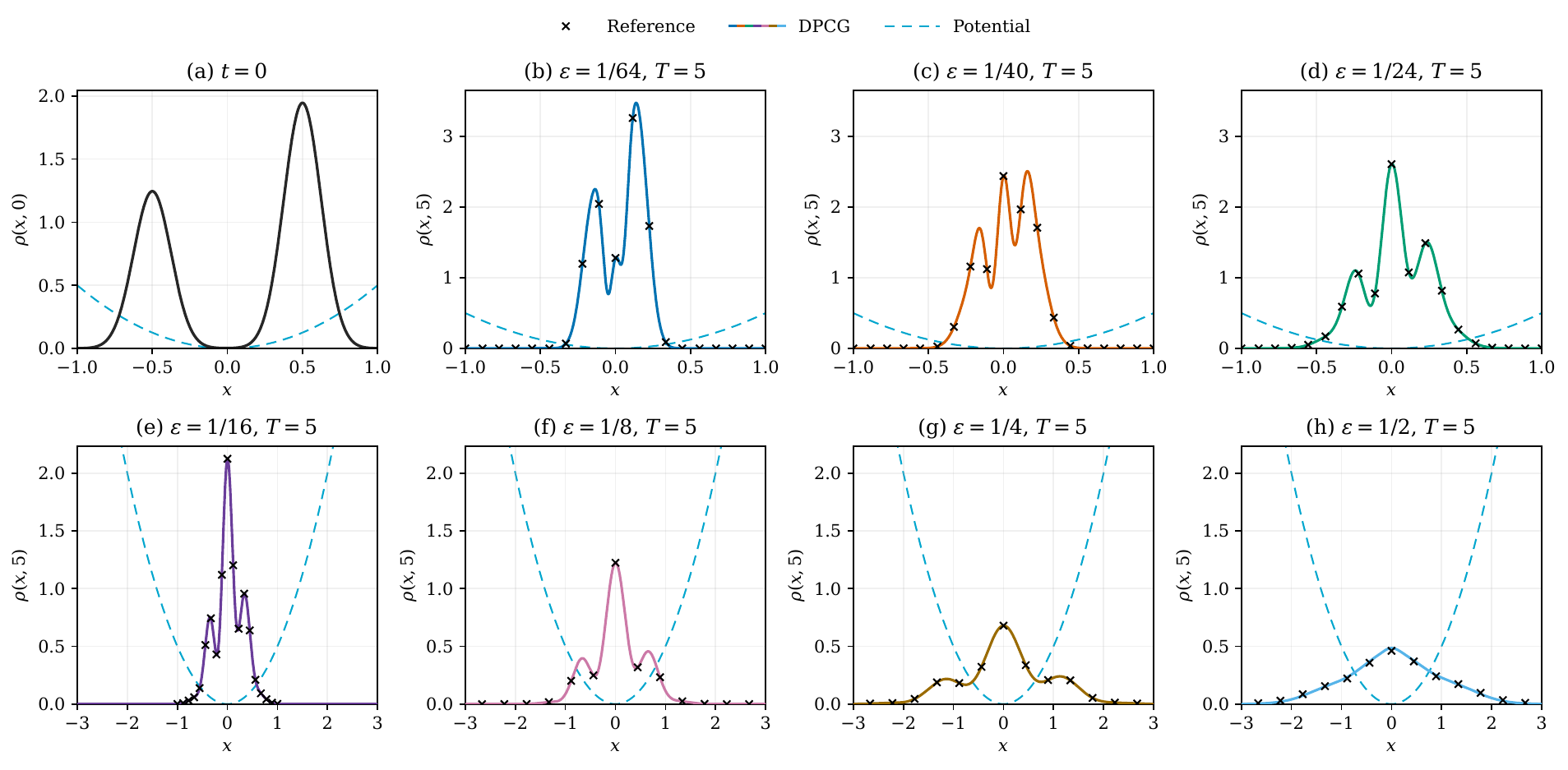}
 \caption{Example~\ref{ex:cl-harmonic}. Density evolution with $\xi=1$
 from the fixed initial state in (a).
 (b)--(h) Densities at $T=5$ for different $\varepsilon$.
 Colored solid curves: DPCG. Black crosses: reference solution.
 Cyan dashed curves: the bare potential $V_0(x)=x^2/2$.}
 \label{fig:epsilon-dependence}
\end{figure}

\begin{figure}[htbp]
 \centering
 \includegraphics[width=0.98\linewidth]
 {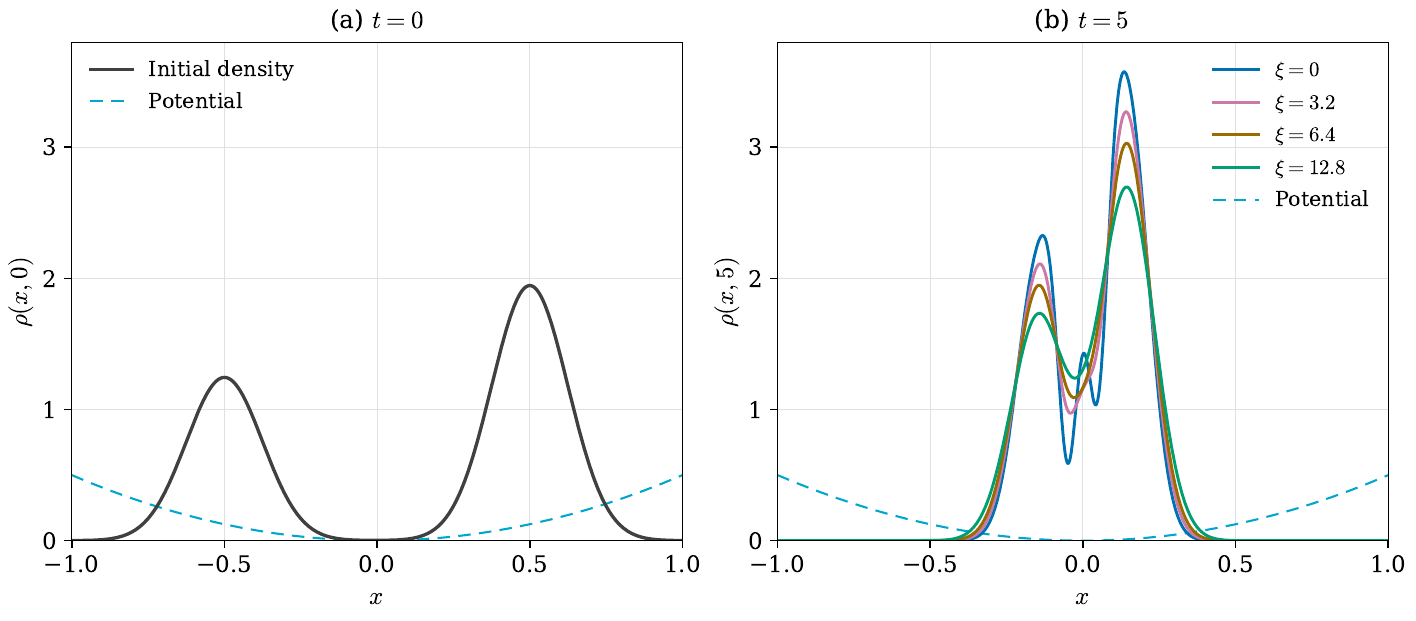}
 \caption{Example~\ref{ex:cl-harmonic}. Density evolution at
 $\varepsilon=1/64$ from the fixed initial state in (a).
 (b) Densities at $T=5$ for different $\xi$.
 Cyan dashed: the bare potential $V_0(x)=x^2/2$.}
 \label{fig:xi-dependence-small}
\end{figure}

\subsection{Convergence tests}
\label{sec:convergence-tests}

\begin{example}[Convergence tests]
\label{ex:cl-convergence}

We examine the low-rank approximation and convergence with respect to
the chaos order, spatial basis size, and time step, using the harmonic
potential $V_0(x)=x^2/2$ and the two-packet initial
state~\eqref{eq:two-peak-state}, with $\xi=1$ and $T=5$.
\begin{itemize}
\item \emph{Low-rank approximation.}
For $\varepsilon=1/64,1/8,1/4,1$, we measure the retained covariance energy
\begin{equation}
 \mathcal E_Q=\frac{\sum_{q=1}^{Q}\lambda_q}{\sum_{q\geq1}\lambda_q},
 \label{eq:kl-retained-energy}
\end{equation}
where $\lambda_q$ are the eigenvalues of the weighted covariance matrix
in Eq.~\eqref{eq:numerical-kernel-factors}.
\item \emph{Chaos order.}
For the same values of $\varepsilon$, we measure the density error
$\|\rho_P(\cdot, T)-\rho^*(\cdot, T)\|_{L^1}$
for $P=0,\ldots,7$, fixing $Q=16$, $K=48$, and $\Delta t=0.025$.

\item \emph{Spatial convergence.}
With $\varepsilon=1$, $Q=16$, $P=4$, and $\Delta t=0.00625$,
we compare solutions for different $K$ with the $K=128$ solution
computed using the same time step and low-rank modes:
\begin{equation}
 E_K^\psi=\norm{\psi_K(\cdot,T)-\psi_{128}(\cdot,T)}_{\mathcal X},\qquad
 E_K^\rho=\norm{\rho_K(\cdot,T)-\rho_{128}(\cdot,T)}_{L^1}.
 \label{eq:spatial-order-errors}
\end{equation}
\item \emph{Temporal convergence.}
We measure successive density differences
\begin{equation}
 D_{\Delta t}=\norm{\rho_{\Delta t}(\cdot,T)
 -\rho_{\Delta t/2}(\cdot,T)}_{L^1},
 \label{eq:rk-richardson}
\end{equation}
for the explicit midpoint rule (RK2) and classical RK4, with
$\varepsilon=1$, $Q=16$, $P=4$, and $K=48$.
The same low-rank modes are used at all time steps.
\end{itemize}
\end{example}

The covariance energy is concentrated in a small number of modes
for all tested values of $\varepsilon$, supporting the use of a low-rank
representation (Figure~\ref{fig:convergence-tests}(a)).
The density error decreases approximately exponentially as the
chaos order increases, before leveling off for the smaller values
of $\varepsilon$ in Figure~\ref{fig:convergence-tests}(b).

Both the wavefunction and density errors decay rapidly as the
spatial basis is enriched, consistent with spectral convergence
of the Hermite approximation
(Figure~\ref{fig:convergence-tests}(c)).
RK2 and RK4 exhibit the expected second- and fourth-order temporal
convergence, respectively
(Figure~\ref{fig:convergence-tests}(d)).

\begin{figure}[htbp]
 \centering
 \includegraphics[width=0.98\linewidth]
 {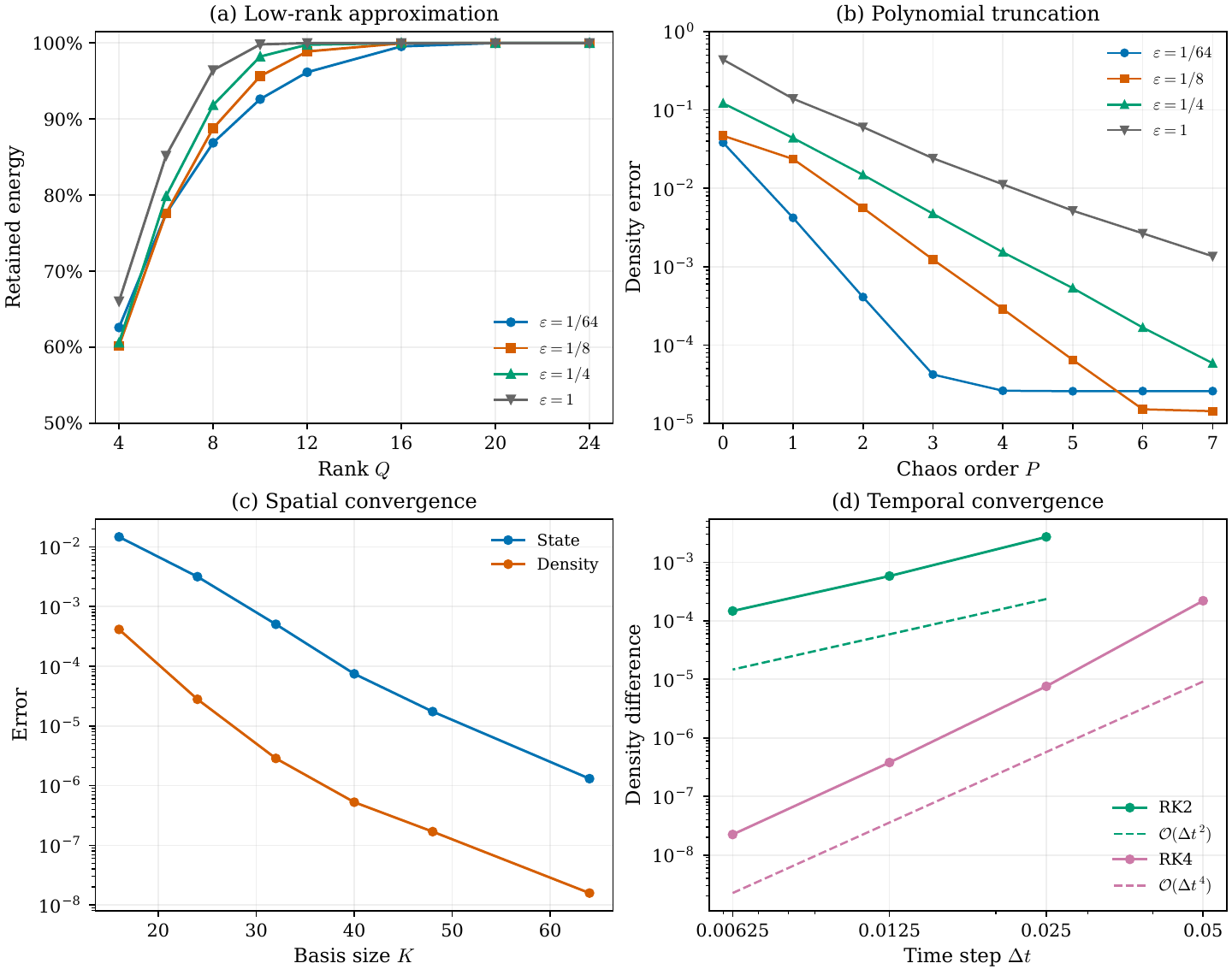}
 \caption{Example~\ref{ex:cl-convergence}. Convergence tests at $T=5$.
 (a) Retained covariance energy.
 (b) Density error versus chaos order.
 (c) Spatial convergence.
 (d) Temporal convergence.}
 \label{fig:convergence-tests}
\end{figure}

\subsection{Computational cost}
\label{sec:cost-comparison}

\begin{example}[Computational cost]
\label{ex:cl-cost}
\label{sec:cost}
\label{sec:representation-cost}

We examine how the computational cost and accuracy of DPCG depend on
its discretization parameters and compare it with existing numerical
methods. We use the harmonic potential $V_0(x)=x^2/2$ and the two-packet
initial state~\eqref{eq:two-peak-state}, with $\varepsilon=1/16$,
$\xi=1$, and the bath parameters in Eq.~\eqref{eq:numerical-parameters}.
We use $T=5$ and $\Delta t=0.025$.
The error is estimated by 
$E_1=\int_{-2}^{2}|\rho(x,5)-\rho^*(x,5)|\,\dd x$,
where $\rho^*$ is the exact reference density~\eqref{eq:exact-harmonic-density}. 
All computations are performed on an AMD Ryzen AI 9 H 365 processor
with single-threaded OpenBLAS.

This example contains two parts:
\begin{itemize}
\item We study the cost and accuracy of DPCG by varying one of $P$,
$Q$, and $K$ at a time about $(P,Q,K)=(4,16,40)$.
The reported times are medians of three runs.

\item We compare the numerical performance of the following six methods:
\begin{itemize}
 \item \textbf{QuAPI}: the quasiadiabatic propagator path integral
 method, widely used for benchmark calculations in open quantum
 dynamics~\cite{ovcharenko2025}.
 It propagates the reduced density matrix by summing system paths
 over a finite memory window~\cite{thorwart2000}.
 \item \textbf{LRD}: the low-rank Dyson method, which uses a low-rank
 approximation of the bath correlation function to factorize the
 Dyson expansion of the reduced density matrix~\cite{zhan2026}.
 \item \textbf{HEOM}: the hierarchical equations of motion, which
 propagate the reduced density matrix together with auxiliary density
 matrices~\cite{tanimura1989,tanimura_numericallyexactapproach_2020}.
 \item \textbf{HOPS}: the linear hierarchy of pure
 states~\cite{suess_hierarchystochasticpure_2014}.
 \item \textbf{LRS}: the low-rank stochastic
 hierarchy~\cite{zhang_lowrankhierarchicalframework_2026}.
 \item \textbf{DPCG}: the deterministic polynomial chaos Galerkin
 method.
\end{itemize}
\end{itemize}
\end{example}

Table~\ref{tab:runtime-scaling} shows that the computation time increases
with $P$ and $Q$, consistent with the growth of the coefficient
count $N_{\mathrm{coef}}=K\binom{Q+P}{P}$, while its dependence on $K$
is approximately linear over the tested range.
The error approaches a plateau in each parameter scan because
the other two parameters and $\Delta t$ remain fixed.

\begin{table}[htbp]
 \centering
 \caption{Example~\ref{ex:cl-cost}. Computation time and error
 of DPCG for different $P$, $Q$, and $K$.}
 \label{tab:runtime-scaling}
\setlength{\tabcolsep}{4pt}
\begin{tabular}{@{}crc@{\hspace{1.5em}}crc@{\hspace{1.5em}}crc@{}}
 \toprule
 \multicolumn{3}{c}{$Q=16$, $K=40$} &
 \multicolumn{3}{c}{$P=4$, $K=40$} &
 \multicolumn{3}{c}{$P=4$, $Q=16$} \\
 \cmidrule(r){1-3}\cmidrule(lr){4-6}\cmidrule(l){7-9}
 $P$ & Time (s) & Error & $Q$ & Time (s) & Error & $K$ & Time (s) & Error \\
 \midrule
 1 & 0.04 & $1.81\times10^{-2}$ & 4 & 0.06 & $1.08\times10^{-2}$ & 16 & 2.41 & $3.65\times10^{-3}$ \\
 2 & 0.09 & $3.48\times10^{-3}$ & 8 & 0.25 & $1.65\times10^{-3}$ & 24 & 3.43 & $1.42\times10^{-4}$ \\
 3 & 0.66 & $5.99\times10^{-4}$ & 12 & 2.02 & $9.00\times10^{-4}$ & 32 & 4.35 & $1.02\times10^{-4}$ \\
 4 & 5.41 & $1.01\times10^{-4}$ & 16 & 5.41 & $1.01\times10^{-4}$ & 40 & 5.41 & $1.01\times10^{-4}$ \\
 5 & 26.27 & $5.16\times10^{-5}$ & 20 & 13.35 & $1.02\times10^{-4}$ & 48 & 6.53 & $1.01\times10^{-4}$ \\
 6 & 98.27 & $4.83\times10^{-5}$ & 24 & 25.46 & $1.02\times10^{-4}$ & 64 & 8.63 & $1.01\times10^{-4}$ \\
 \bottomrule
\end{tabular}

\end{table}

DPCG gives the smallest density error and the shortest computation
time at the settings listed in Table~\ref{tab:method-comparison}.
The number of unknowns in DPCG is the same as in a single HOPS or
LRS trajectory, but the density is obtained from a single deterministic
solve. This avoids both sampling error and the repeated hierarchy
propagations required by HOPS and LRS.
HEOM also avoids trajectory sampling, but propagates a density matrix
for each hierarchy index, requiring $K$ times as many unknowns as DPCG
at the same truncation settings.
The exponential approximation of the bath correlation function introduces 
an additional source of error in both HOPS and HEOM.

QuAPI remains costly even with the smaller spatial basis and short
memory window used here. Its error includes contributions from both
spatial and memory truncation, while refining either increases the
size of its history tensor and hence the computational cost.
With a comparable number of unknowns, LRD yields a larger error and
requires more computation time than DPCG in this test.

\begin{table}[htbp]
 \centering
 \small
 \setlength{\tabcolsep}{8pt}
 \caption{Example~\ref{ex:cl-cost}. Computational cost and error
 of the six methods.}
 \label{tab:method-comparison}
 \begin{tabular}{@{}lrrcl@{}}
  \toprule
  Method & Size & Time (s) & Error & Parameters\\
  \midrule
  \multicolumn{5}{@{}l}{\textit{Density-based formulations}}\\
  QuAPI & $16{,}777{,}216$ & $2156.45$ & $1.63\times10^{-1}$
   & $K=8$, memory time $0.1$\\
  LRD & $43{,}350$ & $11.57$ & $6.85\times10^{-2}$
   & $Q=12$, $2P=8$, $85\times85$ grid\\
  HEOM & $1{,}048{,}320$ & $104.18$ & $5.64\times10^{-3}$
   & $Q=12$, $P=4$, $K=24$\\
  \midrule
  \multicolumn{5}{@{}l}{\textit{Stochastic hierarchies}}\\
  HOPS & $43{,}680\times1000$ & $3079.13$ & $1.08\times10^{-2}$
   & $Q=12$, $P=4$, $K=24$\\
  LRS & $43{,}680\times1000$ & $3320.72$ & $6.94\times10^{-3}$
   & $Q=12$, $P=4$, $K=24$\\
  \midrule
  \multicolumn{5}{@{}l}{\textit{Present method}}\\
  DPCG & $43{,}680$ & $3.07$ & $8.87\times10^{-4}$
   & $Q=12$, $P=4$, $K=24$\\
  \bottomrule
 \end{tabular}
 \par\smallskip
 \begin{minipage}{\textwidth}
  \footnotesize
  \textit{Notes.} Size denotes the number of unknowns.
  For HOPS and LRS, it is reported as the count per trajectory
  $\times 1000$, and the reported times cover the full ensemble.
  For LRD, $2P$ denotes the Dyson interaction order, and the grid
  is used for phase-space quadrature in the Gaussian wave-packet
  approximation.
  We use a smaller $K$ for QuAPI because of the storage and computational
  cost of its history tensor.
 \end{minipage}
\end{table}

\section{Applications}
\label{sec:applications}

We demonstrate the capability of DPCG to solve a range of challenging
problems in open quantum dynamics.
We continue to consider the Caldeira--Leggett model~\eqref{eq:cl-model}, using the bath parameters in
Eq.~\eqref{eq:numerical-parameters}. For multidimensional systems,
each spatial coordinate couples to an independent copy of the bath.
The extended model and the corresponding coefficient equations are given in
Appendix~\ref{app:multidimensional-implementation}.

\subsection{Dynamics in nonquadratic potentials}
\label{sec:nonquadratic-dynamics}
\label{sec:cl-two-dimensional}

\begin{example}[One-dimensional double-well potential]
\label{ex:cl-double-well-1d}
\label{sec:one-dimensional-double-well}
We examine how bath coupling affects the density dynamics in a
nonquadratic double-well potential~\cite{wang2025}. The potential is
\begin{equation}
 V_0(x)=-x^2+2x^4,
 \label{eq:one-dimensional-double-well}
\end{equation}
with the initial state~\eqref{eq:two-peak-state}.
Set $\varepsilon=1/64$ and compare $\xi=0,3.2,12.8$ up to $T=3$.
Spatial discretization uses the scaled Hermite functions
$\varepsilon^{-1/4}\chi_n(x/\sqrt{\varepsilon})$,
with $\chi_n$ defined in Eq.~\eqref{eq:spatial-hermite-functions}.
The parameters are $Q=10$, $P=8$, $K=256$, and $\Delta t=0.01$.
\end{example}

Figure~\ref{fig:one-dimensional-double-well} shows that the two main density
peaks remain near the well minima, while their heights and widths change
with time. The curves are close at $t=1$, and their differences become
clearer at $t=2$ and $3$. Increasing the coupling strength $\xi$ smooths
the peak profiles and suppresses the spatial oscillations between the wells.

\begin{figure}[htbp]
 \centering
 \includegraphics[width=0.98\linewidth]
 {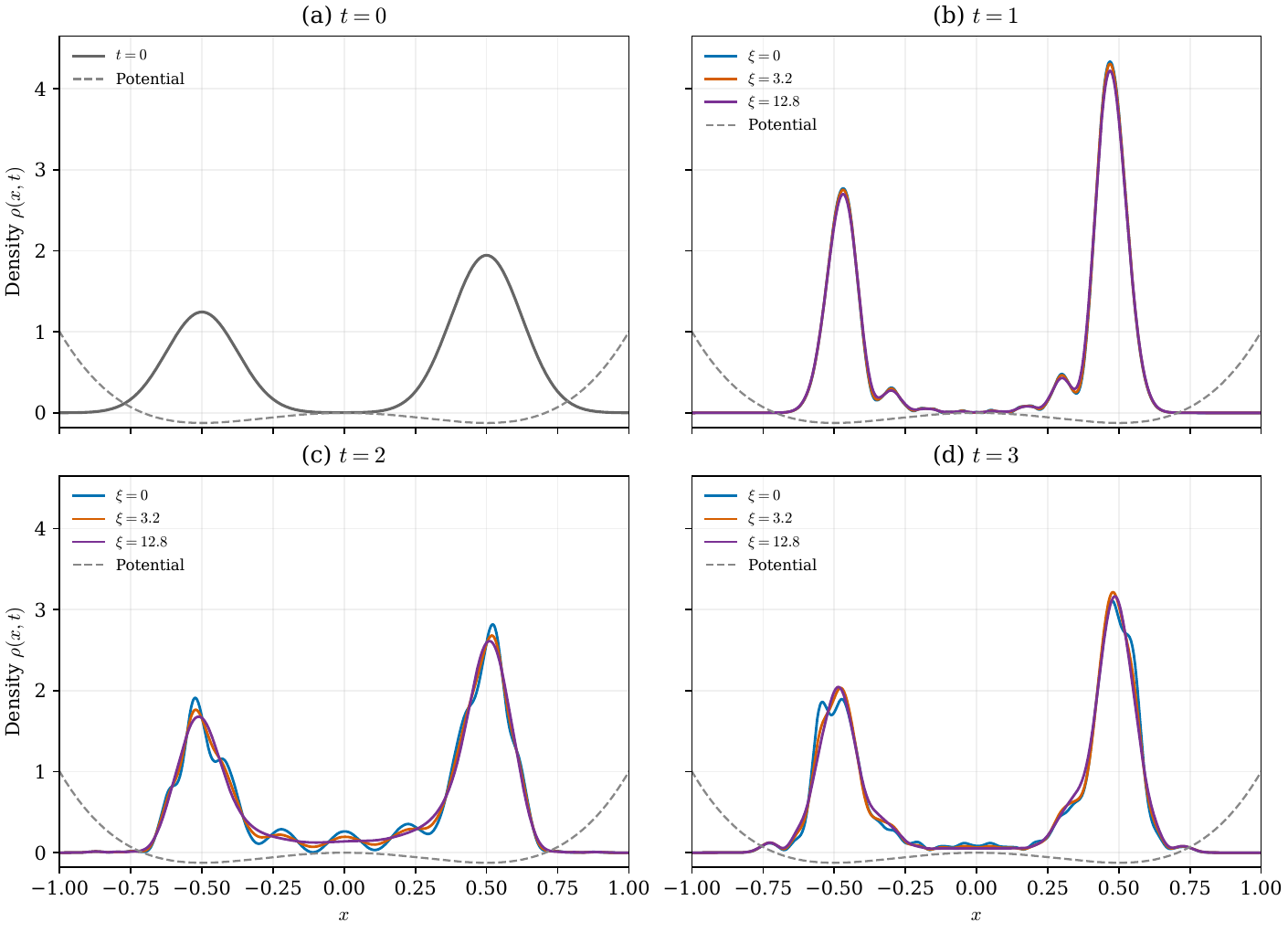}
 \caption{Example~\ref{ex:cl-double-well-1d}. Density evolution in the
 one-dimensional double-well potential. The gray dashed curve shows the bare
 potential $V_0$.}
 \label{fig:one-dimensional-double-well}
\end{figure}

\begin{example}[Two-dimensional double-well potential]
\label{ex:cl-double-well-2d}
\label{sec:two-dimensional-double-well}

We test DPCG on two-dimensional double-well problems with separable and
coupled dynamics. Consider the potentials
\begin{equation}
 \begin{aligned}
 V_{0,\theta}(x_1,x_2)&=-u^2+2u^4+\tfrac12v^2,
 \qquad \theta=0,\ \pi/4,\\
 u&=x_1\cos\theta+x_2\sin\theta,\qquad
 v=-x_1\sin\theta+x_2\cos\theta.
 \end{aligned}
 \label{eq:two-dimensional-double-well}
\end{equation}
Both cases start from
\begin{equation}
 \psi_0^{(2)}(x_1,x_2)=\psi_0(x_1)(\pi\varepsilon)^{-1/4}
 e^{-x_2^2/(2\varepsilon)},
 \label{eq:two-dimensional-double-well-initial}
\end{equation}
where $\psi_0$ is given by Eq.~\eqref{eq:two-peak-state}.
For $\theta=0$, the potential combines a double-well potential in $x_1$
with a harmonic potential in $x_2$.
The two coordinates then evolve independently, while $\theta=\pi/4$
gives coupled dynamics.
We set $\varepsilon=1/32$, $\xi=1$, and $T=3$, and use
$Q=12$ low-rank modes for each coordinate, $P=3$, and $\Delta t=0.01$.
Both cases use the same scaled Hermite basis with $48$ functions in each
coordinate, giving $K=48^2$.
\end{example}

The top row of Figure~\ref{fig:two-dimensional-double-well-density}
shows the independent dynamics for $\theta=0$. The harmonic component
retains a Gaussian profile, while the packets in the double-well
direction alternately narrow and broaden, with corresponding changes
in peak height.
This behavior is consistent with the one-dimensional
results in Example~\ref{ex:cl-double-well-1d}.
For $\theta=\pi/4$, each packet oscillates around its respective well
minimum, with motion both along and across the double-well axis and
changes in width and orientation.

\begin{figure}[htbp]
 \centering
 \includegraphics[width=0.99\linewidth]
 {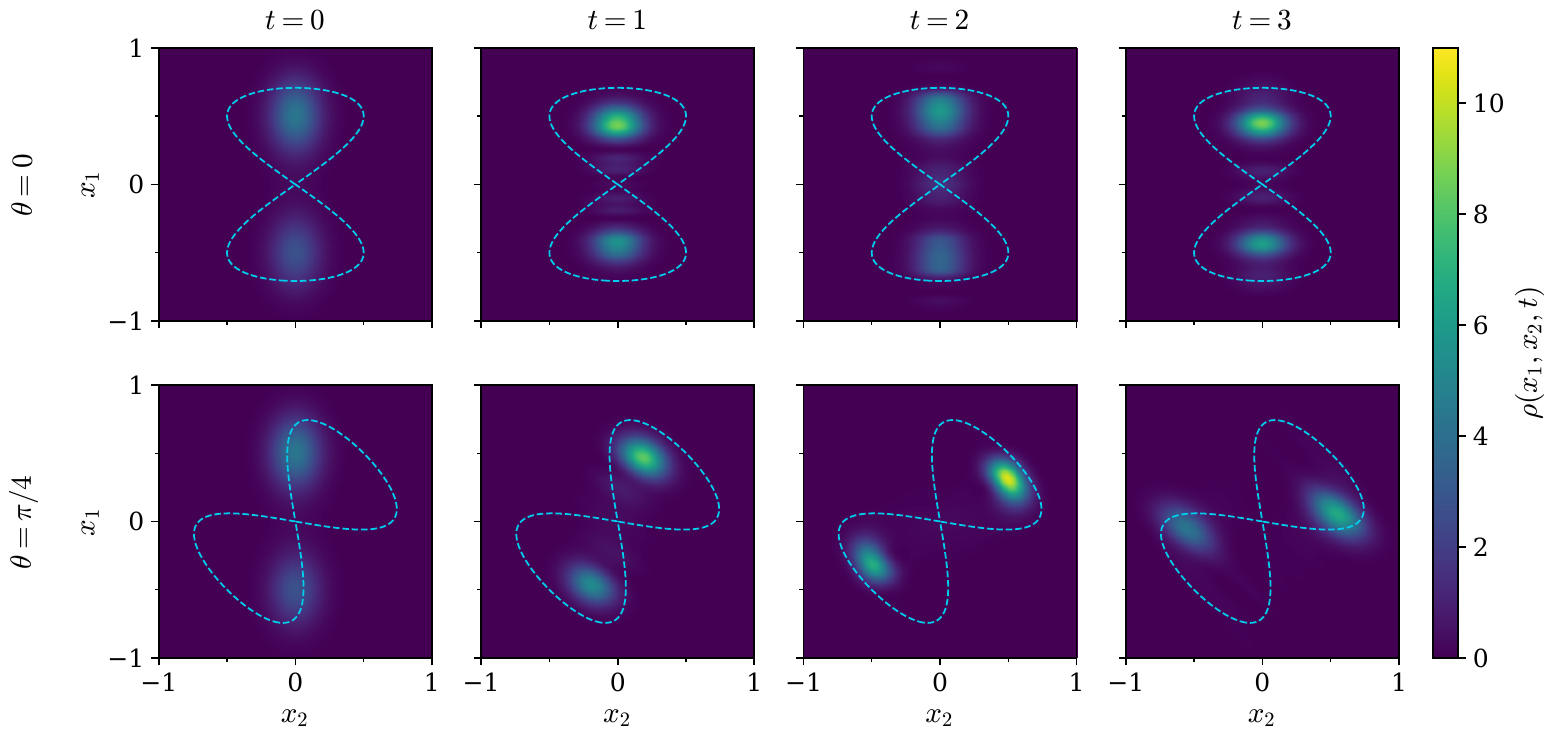}
 \caption{Example~\ref{ex:cl-double-well-2d}. Density evolution from the
 same initial state, with independent dynamics for $\theta=0$ (top)
 and coupled dynamics for
 $\theta=\pi/4$ (bottom).
 Cyan dashed contours indicate $V_{0,\theta}=0$.}
 \label{fig:two-dimensional-double-well-density}
\end{figure}

\subsection{Diffraction and interference}
\label{sec:interference}

\begin{example}[Single- and double-slit diffraction]
\label{ex:cl-single-double-slit}
\label{ex:cl-double-slit}
\label{sec:double-slit}

We study single-slit diffraction and double-slit
interference~\cite{jonsson1961}.
The bare potentials are
\begin{equation}
 \begin{aligned}
 V_{0,j}(x_1,x_2)&=80[1-A_j(x_1)]
 \left[1-f\!\left(\frac{|x_2|-0.05}{0.075}\right)\right],
 \qquad j=1,2,3,\\
 A_1(x)&=g(x-0.3),\qquad A_2(x)=g(x+0.3),\\
 A_3(x)&=A_1(x)+A_2(x),\\
 g(x)&=1-f\!\left(\frac{|x|-0.075}{0.05}\right).
 \end{aligned}
 \label{eq:single-double-slit-potential}
\end{equation}
The two single slits ($j=1,2$) are centered at $x_1=0.3$ and
$-0.3$, respectively, while the double slit ($j=3$) has both openings.
Each slit has a region of width $0.15$ where $V_{0,j}=0$.
The transition function $f$ is
\begin{equation}
 f(s)=\begin{cases}
 0,&s\leq0,\\
 10s^3-15s^4+6s^5,&0<s<1,\\
 1,&s\geq1.
 \end{cases}
 \label{eq:slit-smoothstep}
\end{equation}

The initial state is the normalized Gaussian packet
\begin{equation}
 \psi_0(x_1,x_2)=\frac{1}{\sqrt{2\pi\sigma_1\sigma_2}}
 \exp\!\left[-\frac{x_1^2}{4\sigma_1^2}
 -\frac{(x_2+0.8)^2}{4\sigma_2^2}
 +\frac{i8x_2}{\varepsilon}\right],
 \label{eq:single-double-slit-initial}
\end{equation}
where $\sigma_1=0.8$ and $\sigma_2=0.18$.
The packet moves toward the barrier with mean momentum $8$ in the
$x_2$ direction. We set $\varepsilon=1/8$ and $T=0.4$.

Spatial discretization uses a Fourier basis with
$K=384\times480$ on $[-4,4)\times[-5,5)$.
Each bath uses $Q=2$ low-rank modes.
Time integration uses Strang
splitting~\cite{strang_constructioncomparisondifference_1968}
with $\Delta t=10^{-3}$.

This example includes two tests:
\begin{itemize}
    \item Comparison of the two single slits and
    the double slit at $\xi=1$, with $\beta=5$ and $P=2$.
    \item Dependence of double-slit interference on the bath
    coupling strength and temperature, using $\xi=1,64,128,256$
    at $\beta=5$ and $\beta=5,2,1,0.5$ at $\xi=64$.
    The chaos degree is $P=4$.
\end{itemize}

\end{example}

The DPCG results in
Figure~\ref{fig:single-double-slit-density} reproduce the
single-slit diffraction and double-slit interference observed
experimentally~\cite{bach2013,zeilinger1988,tavabi2019}. Each single slit produces a broad
diffraction pattern, while the double-slit density is not the sum of
the two single-slit densities but exhibits alternating bright and dark
interference fringes.

\begin{figure}[!htbp]
 \centering
 \includegraphics[width=\linewidth]
 {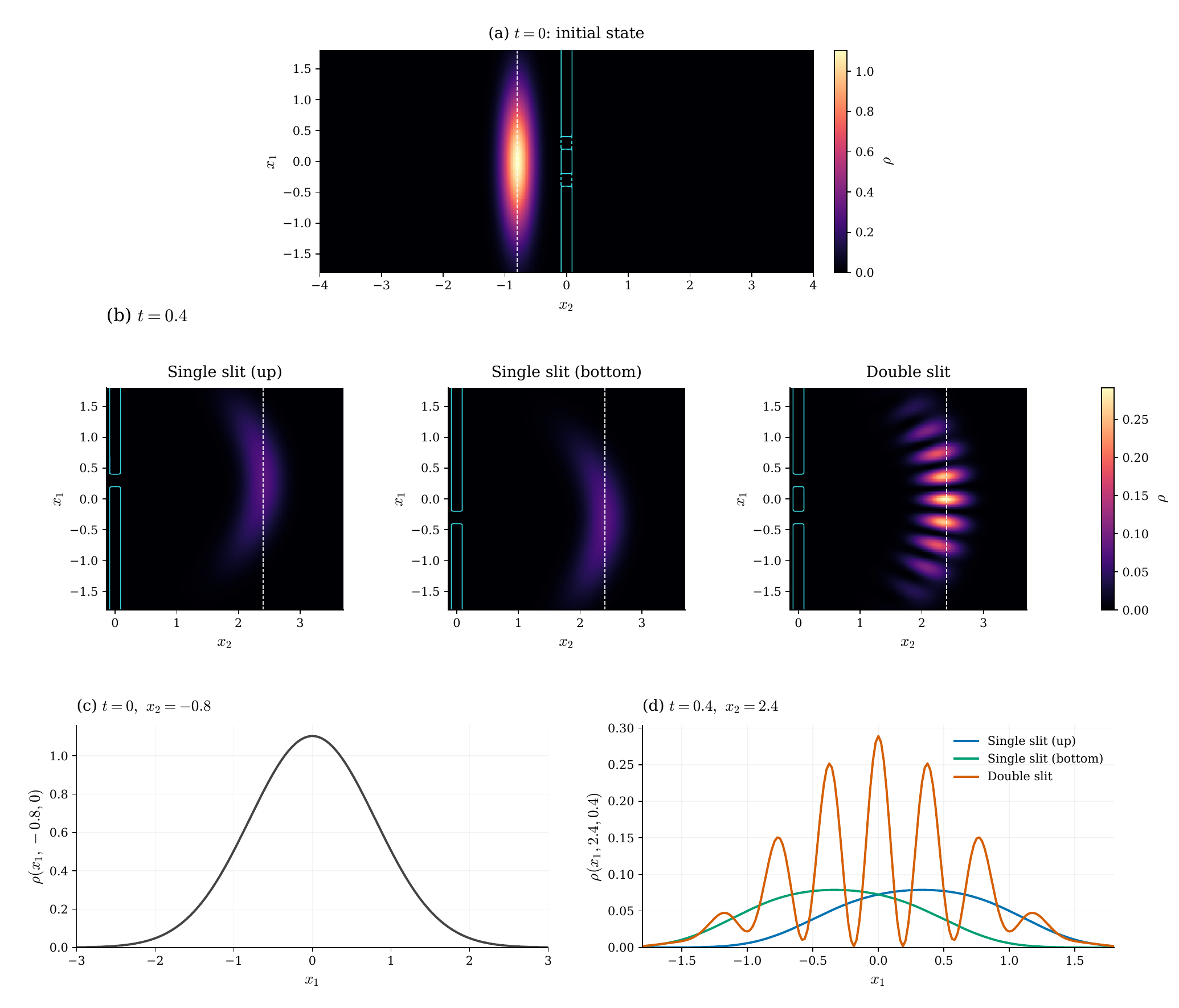}
 \caption{Example~\ref{ex:cl-single-double-slit}. Single- and double-slit
 diffraction at $\xi=1$. (a) Initial density. (b) Transmitted densities.
 (c,d) Cross sections along $x_2=-0.8$ and $x_2=2.4$ in (a,b), respectively.
 Solid cyan contours at $V_{0,j}=40$ outline the barrier and its slits;
 dashed cyan boxes in (a) mark the two openings.}
 \label{fig:single-double-slit-density}
\end{figure}

Figure~\ref{fig:double-slit-density} compares the transmitted
densities for different bath coupling strengths and temperatures.
At fixed $\beta=5$, increasing the bath coupling shifts the
transmitted packet toward smaller $x_2$.
This is because the counterterm in the potential $V(x)$ increases with $\xi$,
strengthening the harmonic confinement.
The diffraction peaks observed at $x_2=2.4$ become lower,
most noticeably at $\xi=256$.

At fixed $\xi=64$, as the bath temperature increases ($\beta$ decreases), the bright
and dark fringes become less distinct, while the transmitted packet
remains at approximately the same position.
The interference fringes observed at $x_2=2.4$ show lower peaks
and higher minima as $\beta$ decreases, resulting in a progressively
smoother overall density profile.
These results show that double-slit interference is suppressed
as the influence of the bath becomes stronger.

\begin{figure}[!htbp]
 \centering
 \includegraphics[width=\linewidth]
 {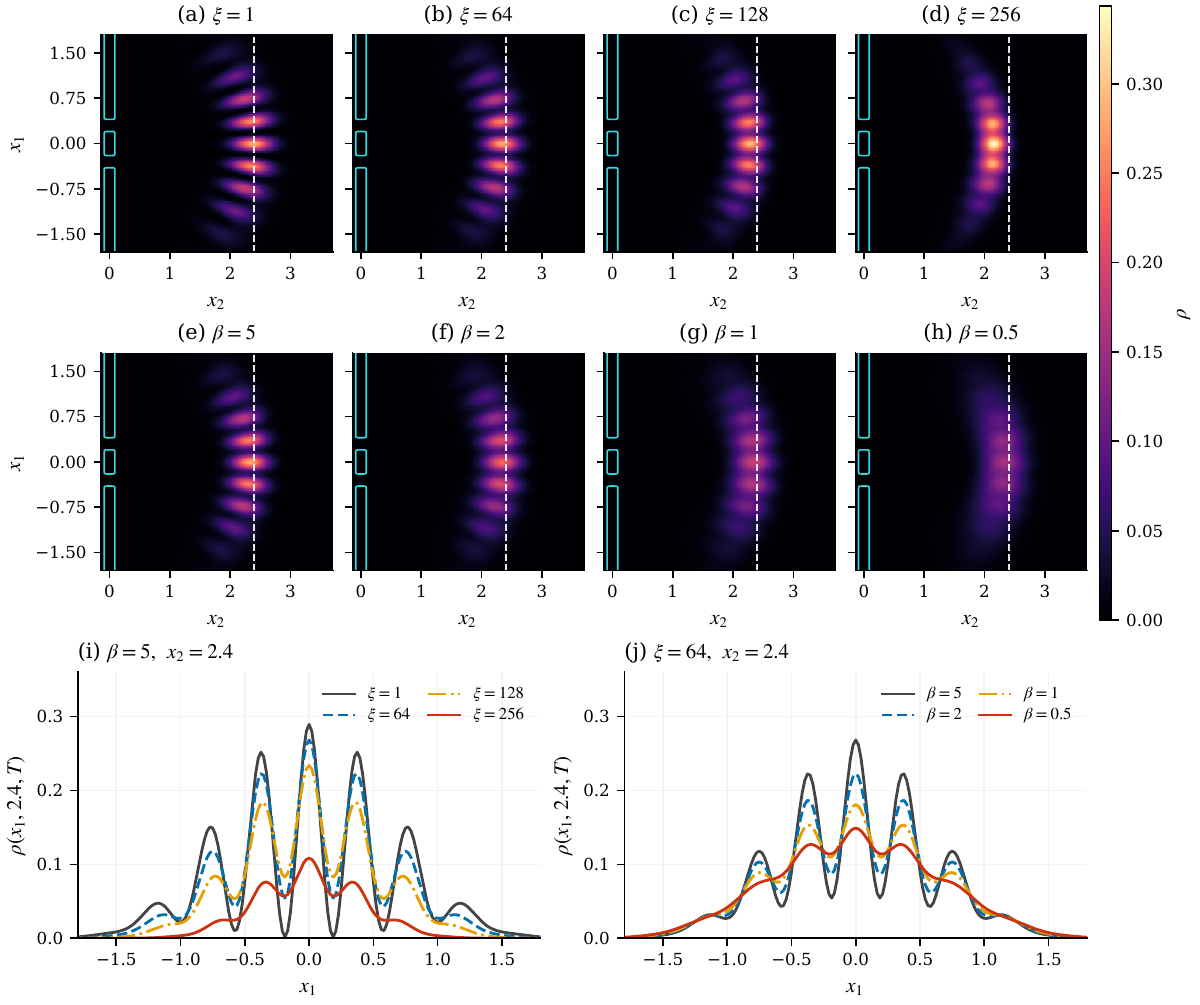}
 \caption{Example~\ref{ex:cl-single-double-slit}. Effects of bath coupling
 strength and temperature on double-slit interference at $T=0.4$.
 (a--d) Transmitted densities for $\xi=1,64,128,256$ at $\beta=5$.
 (e--h) Transmitted densities for $\beta=5,2,1,0.5$ at $\xi=64$.
 (i,j) Corresponding cross sections at $x_2=2.4$ for
 (a--d) and (e--h), respectively.
 Cyan contours at $V_{0,3}=40$ outline the barrier
 and its two slits.}
 \label{fig:double-slit-density}
\end{figure}

\subsection{Coordinate correlations in three dimensions}
\label{sec:three-dimensional-coupled}

\begin{example}[Three coupled harmonic oscillators]
\label{ex:cl-coupled-3d}

We consider a chain of three coupled harmonic oscillators,
described by Eq.~\eqref{eq:cl-model-d} with the bare potential
\begin{equation}
 V_0(\bm x)=\frac12\bm x^T A\bm x,
 \qquad
 A=\begin{pmatrix}
 1&0.3&0\\0.3&1.44&0.3\\0&0.3&1.96
 \end{pmatrix},
 \qquad \bm x\in\mathbb R^3.
 \label{eq:three-dimensional-potential}
\end{equation}
Set $\varepsilon=1/8$, $\xi=1$, and $T=3$, and take the initial state
\begin{equation}
 \psi_0(\bm x)=\frac{(\det D_0)^{1/4}}{(\pi\varepsilon)^{3/4}}
 \exp\!\left[-\frac{(\bm x-\bm a)^TD_0(\bm x-\bm a)}{2\varepsilon}\right],
 \quad D_0=\operatorname{diag}(1,1.5,2),\quad
 \bm a=(1/2,0,0)^T.
 \label{eq:three-dimensional-initial}
\end{equation}

The Gaussian initial system and bath states remain Gaussian under the
quadratic total Hamiltonian~\cite{weedbrook2012}. The reference position density therefore
has the form
\begin{equation}
 \rho^{(3),*}(\bm x,t)
 =\frac{\exp[-\tfrac12(\bm x-\bm\mu(t))^T\Sigma_x(t)^{-1}(\bm x-\bm\mu(t))]}
 {(2\pi)^{3/2}\sqrt{\det\Sigma_x(t)}}.
 \label{eq:three-dimensional-exact-density}
\end{equation}
Here $\bm\mu(t)=\langle\bm x\rangle_t$ and
$\Sigma_x(t)=\langle\bm x\bm x^T\rangle_t-\bm\mu(t)\bm\mu(t)^T$
are the mean vector and covariance matrix of the three system coordinates
in the reference solution. Their explicit expressions are given in
Appendix~\ref{app:multidimensional-implementation}.

The scaled Hermite functions from Section~\ref{sec:nonquadratic-dynamics}
are used in each coordinate, with the cutoff $n_1+n_2+n_3\leq14$
giving $K=680$ spatial functions.
Set $Q=12$ for each bath, $P=3$, and $\Delta t=0.025$.
\end{example}

\begin{figure}[htbp]
 \centering
 \includegraphics[width=0.99\linewidth]
 {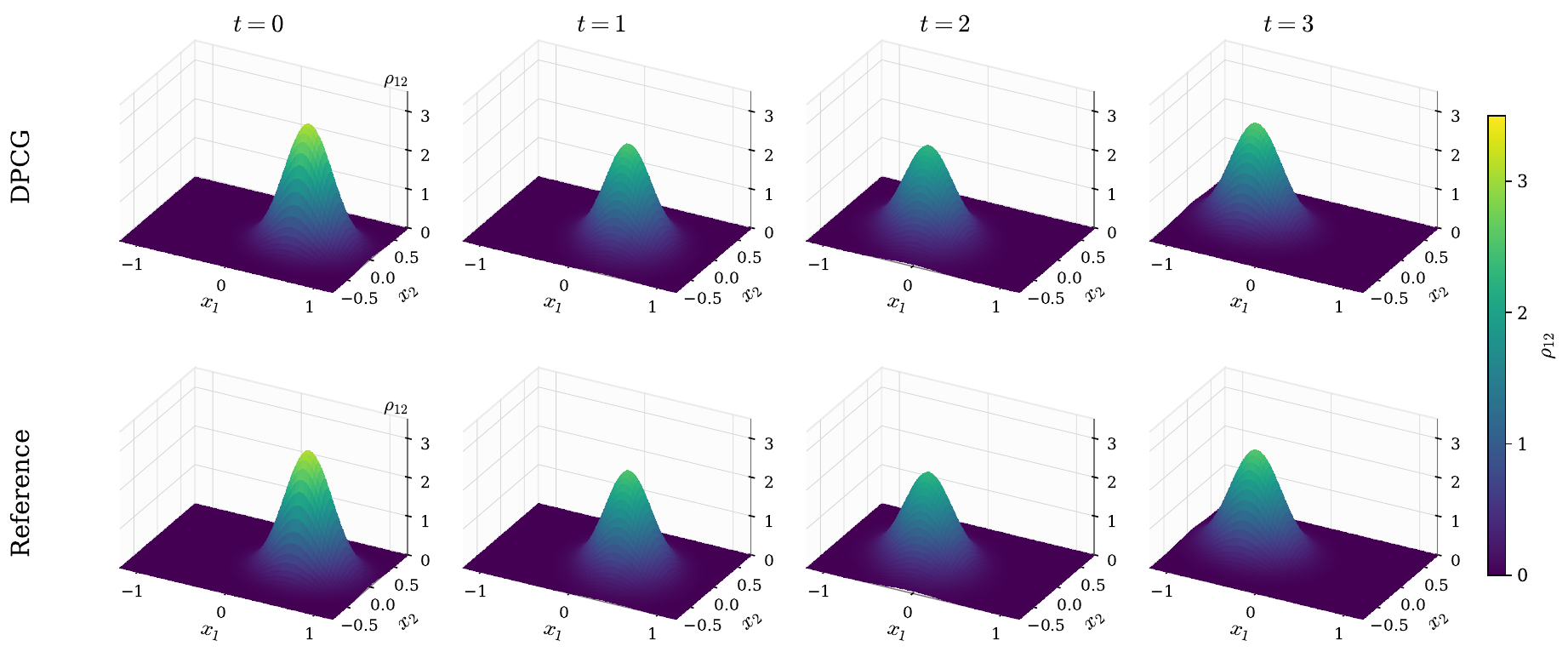}
 \caption{Example~\ref{ex:cl-coupled-3d}. DPCG (top) and reference
 (bottom) marginal densities $\rho_{12}=\int_{\mathbb R}\rho\,\dd x_3$.}
 \label{fig:three-dimensional-marginals}
\end{figure}

The initially independent coordinates develop correlations during the
evolution, as measured by the covariances
$C_{ij}(t)=\langle x_ix_j\rangle_t-\langle x_i\rangle_t\langle x_j\rangle_t$
in Figure~\ref{fig:three-dimensional-correlations}.
Correlations arise between the directly coupled pairs and also between
$x_1$ and $x_3$: although $A_{13}=0$, coupling through $x_2$ generates
a nonzero $C_{13}$.
Resolving these correlations in coupled three-dimensional dynamics
with non-Markovian bath memory poses a substantial computational
challenge for general numerical methods.
DPCG captures the growth and subsequent evolution
of these correlations, including changes in sign, in close agreement
with the reference solution. The marginal density
$\rho_{12}=\int\rho\,\dd x_3$ in
Figure~\ref{fig:three-dimensional-marginals} also closely reproduces
the position, width, and peak height of the reference distribution.

\begin{figure}[htbp]
 \centering
 \includegraphics[width=0.78\linewidth]
 {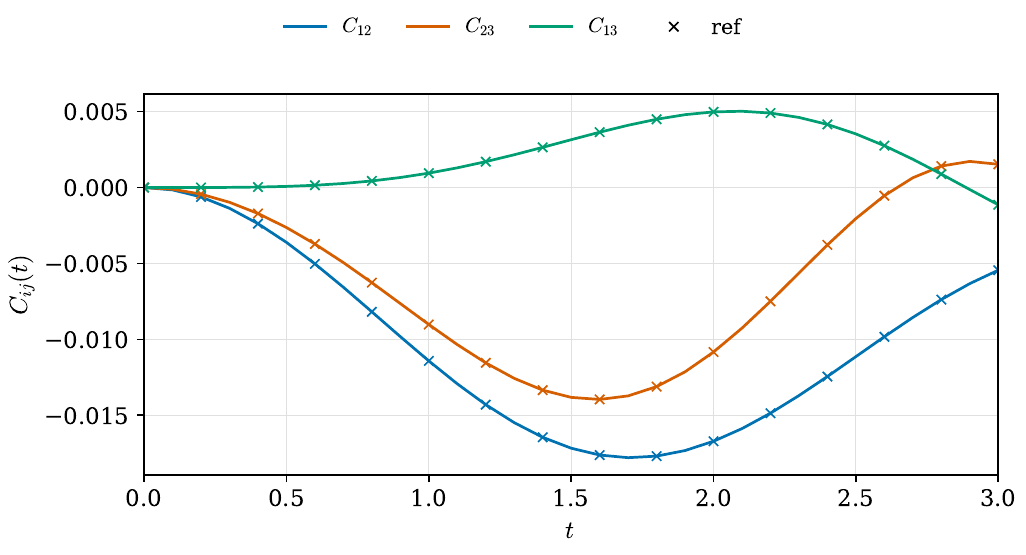}
 \caption{Example~\ref{ex:cl-coupled-3d}. Coordinate cross covariances.
 Crosses show the reference solution, and solid curves
 show the DPCG solution.}
 \label{fig:three-dimensional-correlations}
\end{figure}

\section{Conclusions}
\label{sec:conclusion}

We have developed and analyzed DPCG as a simple and general method
for non-Markovian open quantum dynamics. The method reduces NMQSD
to a closed system of deterministic coefficient equations that can
be solved with standard numerical methods. A single deterministic solve
of the same size as one HOPS trajectory yields the reduced density
matrix, eliminating statistical sampling error and reducing the
computational cost by approximately a factor of $N_{\mathrm{traj}}$.

The numerical experiments demonstrate the ease of use, accuracy,
and computational efficiency of DPCG. The method can handle challenging
non-Markovian problems using standard numerical solvers. In the
benchmarks, it achieves smaller errors at lower computational cost
than the tested stochastic hierarchy methods.

A remaining challenge is the growth of the polynomial chaos space with
the number of retained bath modes and the expansion order. Adaptive
selection of the chaos basis is a natural direction for reducing this
cost and extending the method to larger systems.

\begin{acknowledgments}
This work was supported by the Academic Research Fund of the Ministry of Education of Singapore under Grant No.~A-8002392-00-00.
\end{acknowledgments}

\appendix
\counterwithin*{equation}{section}
\section{Error analysis and trace conservation}
\label{sec:convergence}
\label{app:convergence}

We prove Theorem~\ref{thm:main-error} by separating the state error
into low-rank truncation, Galerkin truncation, and numerical coefficient
error.

\subsection{State space, assumptions, and error decomposition}
\label{sec:convergence-preliminaries}

Let $(\Omega,\mathcal F,\mathbb P)$ carry independent standard complex
Gaussian variables $\eta_q\sim\CN(0,1)$, $q\geq1$.
Define $\mathcal X$ as the closure in $L^2(\Omega;\mathcal H)$ of
$\mathcal H$-valued holomorphic polynomials in finitely many
$\eta_q^*$~\cite{bargmann1961,hall2000}.
For $v,w\in\mathcal X$, the norm and inner product are
\begin{equation}
 \norm{v}_{\mathcal X}^2=\E\norm{v}_{\mathcal H}^2,\qquad
 \langle v,w\rangle_{\mathcal X}
 =\E\langle v,w\rangle_{\mathcal H}.
 \label{eq:gaussian-state-norm}
\end{equation}
The inner product is conjugate-linear in its first argument.
Let $\mathcal X_Q\subset\mathcal X$ be the closed subspace generated
by polynomials in $\eta_1^*,\ldots,\eta_Q^*$.

The kernel $\alpha$ is continuous, Hermitian, and positive semidefinite
on $[0,T]^2$. We use the factors
$b_q=\sqrt{\lambda_q}\,\varphi_q$ and the first-$Q$-mode truncations
$\widehat\alpha$, $\widehat z$ of Section~\ref{sec:low-rank-modes}.

\begin{assumption}[Continuous evolution]
\label{ass:convergence}
The following hold on $[0,T]$.
\begin{enumerate}
 \item $H$ and $L$ are bounded operators on $\mathcal H$, $H$ is
 self-adjoint, and $\norm{\psi_0}_{\mathcal H}=1$.
 \item The state $\psi_t$ is holomorphic in the Gaussian coordinates.
 Its directional derivatives exist in $\mathcal X$, are linear in $c$,
 and satisfy
 \[
  \norm{\sum_{q\geq1}c_q\partial_{\eta_q^*}\psi_t}_{\mathcal X}
  \leq C_t\norm{c}_{\ell^2},\qquad c\in\ell^2,
 \]
 for some $C_t<\infty$ independent of $c$.
 The same assumption holds for each finite-mode equation.
 \item The full Gaussian equation admits a strongly continuous unitary
 propagator $\mathcal U(t,s)$ satisfying
 \begin{equation}
  \norm{\mathcal U(t,s)v}_{\mathcal X}=\norm{v}_{\mathcal X},
  \qquad v\in\mathcal X,\quad 0\leq s\leq t\leq T.
  \label{eq:gaussian-propagator-stability}
 \end{equation}
 The same assumption holds for the finite-mode and polynomial Galerkin
 equations. The error equations below satisfy the
 variation-of-constants formula in $\mathcal X$.
\end{enumerate}
\end{assumption}

Here $\norm{L}$ denotes the operator norm on $\mathcal H$.
Write $\mathcal U_Q(t,s)$ for the finite-mode propagator.
In particular, the exact, finite-mode, and polynomial Galerkin
solutions all have unit $\mathcal X$ norm.

In this section, $x$ denotes the system coordinate, or the component
index for a finite-dimensional system.
Using the notation of Section~\ref{sec:general-projection}, we denote
the exact, finite-mode, and Galerkin states by
\begin{equation}
 \psi_t=\psi(\cdot,t;z),\qquad
 \psi_t^Q=\psi(\cdot,t;\widehat z),\qquad
 \psi_t^{Q,P}=\psi_m(\cdot,t;\widehat z).
 \label{eq:state-notation}
\end{equation}
The finite-mode solution uses both $\widehat z$ and $\widehat\alpha$.
The Galerkin solution has the expansion
\[
 \psi_t^{Q,P}(x)=\sum_{\bm k\in\mathcal K_{Q,P}}
 u_{\bm k}(x,t)\mathcal P_{\bm k},\qquad
 u_{\bm k}(\cdot,t)\in\mathcal H,
\]
with coefficients satisfying Eq.~\eqref{eq:chaos-hierarchy} and initial
conditions $u_{\bm0}(x,0)=\psi_0(x)$ and $u_{\bm k}(x,0)=0$
for $\bm k\ne\bm0$.
Both approximations precede spatial and time discretization.

At $t_n=n\Delta t$, $0\leq n\leq N_T$ and
$N_T\Delta t=T$, let
$\widehat u_{\bm k}^n(\cdot)\in\mathcal H$ be the numerical
approximations to $u_{\bm k}(\cdot,t_n)$. Define the numerical state
$\psi^{Q,P,n}$ at $t_n$ by
\begin{equation}
 \psi^{Q,P,n}(x)
 =\sum_{\bm k\in\mathcal K_{Q,P}}
 \widehat u_{\bm k}^n(x)\mathcal P_{\bm k}.
 \label{eq:numerical-state-reconstruction}
\end{equation}

Define the discarded eigenvalue sum and the coupling bound by
\begin{equation}
 \tau_Q:=\sum_{q>Q}\lambda_q,\qquad
 A_T:=T\norm{L}^2\int_0^T\alpha(s,s)\,\dd s.
 \label{eq:polynomial-coupling-bound}
\end{equation}

\subsection{Low-rank approximation}
\label{sec:low-rank-error}

Regard $\psi_t^Q\in\mathcal X_Q$ as an element of $\mathcal X$
independent of $\eta_q$ for $q>Q$. The exact and truncated states
are then compared on the same probability space.

Under Assumption~\ref{ass:convergence}, the kernel expansion and the
chain rule used in Eq.~\eqref{eq:chain-rule} give the coordinate forms
of the memory operators:
\begin{equation}
 D_t=\sum_{q\geq1}b_q(t)\partial_{\eta_q^*},\qquad
 D_t^Q=\sum_{q=1}^Qb_q(t)\partial_{\eta_q^*}.
 \label{eq:common-memory-operators}
\end{equation}
On the respective solutions,
\begin{equation}
 D_t\psi_t=\int_0^t\alpha(t,s)
     \frac{\delta\psi_t}{\delta z_s^*}\,\dd s,\qquad
 D_t^Q\psi_t^Q=\int_0^t\widehat\alpha(t,s)
     \frac{\delta\psi_t^Q}{\delta\widehat z_s^*}\,\dd s.
 \label{eq:two-memory-actions}
\end{equation}
The two equations become
\begin{equation}
 \partial_t\psi_t=G_t\psi_t,\qquad
 \partial_t\psi_t^Q=G_t^Q\psi_t^Q,\qquad \psi_0^Q=\psi_0,
 \label{eq:full-and-truncated-evolution}
\end{equation}
where
\begin{equation}
 G_t=-iH+Lz_t^*-L^\dagger D_t,\qquad
 G_t^Q=-iH+L\widehat z_t^*-L^\dagger D_t^Q.
 \label{eq:full-and-truncated-generators}
\end{equation}
Define
\begin{equation}
 \delta_Q(t):=\norm{(b_q(t))_{q>Q}}_{\ell^2}.
 \label{eq:kl-tail}
\end{equation}
Independence of the Gaussian variables gives
$\E|z_t-\widehat z_t|^2
=\delta_Q(t)^2=\alpha(t,t)-\widehat\alpha(t,t)$.
Since the summands $|b_q(t)|^2$ are nonnegative, Tonelli's theorem
allows us to interchange the sum and the integral.
Using $b_q=\sqrt{\lambda_q}\,\varphi_q$ and
$\norm{\varphi_q}_{L^2(0,T)}=1$, we obtain
\begin{equation}
 \begin{aligned}
 \int_0^T\delta_Q(t)^2\,\dd t
 &=\int_0^T\sum_{q>Q}|b_q(t)|^2\,\dd t
 =\sum_{q>Q}\int_0^T|b_q(t)|^2\,\dd t\\
 &=\sum_{q>Q}\lambda_q\int_0^T|\varphi_q(t)|^2\,\dd t
 =\sum_{q>Q}\lambda_q=\tau_Q.
 \end{aligned}
 \label{eq:kl-tail-identities}
\end{equation}

\begin{theorem}[Low-rank truncation error]
\label{thm:kl-error}
Under Assumption~\ref{ass:convergence},
\begin{equation}
 \sup_{0\leq t\leq T}\norm{\psi_t-\psi_t^Q}_{\mathcal X}^2
 \leq T\norm{L}^2\tau_Q.
 \label{eq:kl-mean-square-error}
\end{equation}
Hence $\psi^Q\to\psi$ in $C([0,T];\mathcal X)$ as $Q\to\infty$.
\end{theorem}

\begin{proof}
For $e_t=\psi_t-\psi_t^Q$, subtraction gives
\begin{equation}
 \partial_te_t=G_te_t+(G_t-G_t^Q)\psi_t^Q,\qquad e_0=0.
 \label{eq:low-rank-error-evolution}
\end{equation}
Because $\psi_t^Q$ is independent of the omitted coordinates,
\begin{equation}
 (D_t-D_t^Q)\psi_t^Q=0,\qquad
 (G_t-G_t^Q)\psi_t^Q=L(z_t^*-\widehat z_t^*)\psi_t^Q.
 \label{eq:omitted-memory-action}
\end{equation}
The process tail depends only on the omitted variables and is
independent of $L\psi_t^Q$. Consequently,
\begin{equation}
 \norm{(G_t-G_t^Q)\psi_t^Q}_{\mathcal X}^2
 =\E|z_t-\widehat z_t|^2\,
    \E\norm{L\psi_t^Q}_{\mathcal H}^2
 =\delta_Q(t)^2\norm{L\psi_t^Q}_{\mathcal X}^2.
 \label{eq:low-rank-forcing-norm}
\end{equation}
Since $\norm{\psi_t^Q}_{\mathcal X}=1$, the forcing norm is bounded by
$\norm{L}\delta_Q(t)$. It is integrable in time because
$\delta_Q\in L^2(0,T)$.
The variation-of-constants formula in Assumption~\ref{ass:convergence}
applied to Eq.~\eqref{eq:low-rank-error-evolution} gives
\[
 e_t=\mathcal U(t,0)e_0+
 \int_0^t\mathcal U(t,s)(G_s-G_s^Q)\psi_s^Q\,\dd s.
\]
Using $e_0=0$ and Eq.~\eqref{eq:omitted-memory-action}, we obtain
\mbox{$e_t=\int_0^t\mathcal U(t,s)
L(z_s^*-\widehat z_s^*)\psi_s^Q\,\dd s$}.
Unitarity of $\mathcal U(t,s)$ and
Eq.~\eqref{eq:low-rank-forcing-norm} then yield
\begin{equation}
 \begin{aligned}
 \norm{e_t}_{\mathcal X}
 &\leq\int_0^t
       \norm{L(z_s^*-\widehat z_s^*)\psi_s^Q}_{\mathcal X}\,\dd s\\
  &\leq \norm{L}\int_0^t\delta_Q(s)\,\dd s
  \leq \norm{L}\sqrt{t}
        \left(\int_0^t\delta_Q(s)^2\,\dd s\right)^{1/2}\\
  &\leq\norm{L}\sqrt T
        \left(\int_0^T\delta_Q(s)^2\,\dd s\right)^{1/2}
   =\norm{L}\sqrt{T\tau_Q}.
 \end{aligned}
 \label{eq:low-rank-duhamel}
\end{equation}
The last line uses $t\leq T$ and Eq.~\eqref{eq:kl-tail-identities}.
Squaring and taking the supremum over $t\in[0,T]$ proves the bound.
The solutions are continuous in $\mathcal X$ by
Assumption~\ref{ass:convergence}, and the trace-class property gives
$\tau_Q\to0$, proving the stated convergence.
\end{proof}

\subsection{Polynomial Galerkin approximation}
\label{sec:polynomial-error}

For fixed $Q$, let $\Pi_P$ be the orthogonal projection from
$\mathcal X_Q$ onto the polynomial subspace
\begin{equation}
 \mathcal X_{Q,P}
 =\left\{\sum_{\bm k\in\mathcal K_{Q,P}}
 v_{\bm k}\mathcal P_{\bm k}:v_{\bm k}\in\mathcal H\right\}.
 \label{eq:finite-polynomial-space}
\end{equation}
The coefficient system in Eq.~\eqref{eq:chaos-hierarchy} can be written
in operator form as
\begin{equation}
 \partial_t\psi_t^{Q,P}
 =\Pi_PG_t^Q\psi_t^{Q,P},\qquad
 \psi_0^{Q,P}=\psi_0.
 \label{eq:polynomial-galerkin-evolution}
\end{equation}

\begin{theorem}[Polynomial Galerkin error]
\label{thm:chaos-error}
For any fixed $Q$ and $T$, Assumption~\ref{ass:convergence} implies
\begin{equation}
 \sup_{0\leq t\leq T}\norm{\psi_t^Q-\psi_t^{Q,P}}_{\mathcal X}^2
 \leq e^{A_T}\frac{A_T^{P+1}}{(P+1)!},
 \qquad P\geq0.
 \label{eq:polynomial-mean-square-error}
\end{equation}
This gives factorial decay of the mean-square error bound in $P$.
\end{theorem}

\begin{proof}
Set $e_t=\psi_t^Q-\psi_t^{Q,P}$.
Write $\psi_t^{Q,P}=\sum_{i=0}^P\psi_{t,i}^{Q,P}$, where
$\psi_{t,i}^{Q,P}=(\Pi_i-\Pi_{i-1})\psi_t^{Q,P}$
is the homogeneous component of degree $i$, with $\Pi_{-1}=0$.

Subtracting Eq.~\eqref{eq:polynomial-galerkin-evolution} from the
finite-mode equation in Eq.~\eqref{eq:full-and-truncated-evolution} gives
\[
 \partial_te_t=G_t^Qe_t+R_t,\qquad
 R_t:=(I-\Pi_P)G_t^Q\psi_t^{Q,P},\qquad e_0=0.
\]
Since $H\psi_t^{Q,P},L^\dagger D_t^Q\psi_t^{Q,P}\in\mathcal X_{Q,P}$,
expanding $G_t^Q$ gives
\begin{equation}
 \begin{aligned}
 R_t&=(I-\Pi_P)\bigl(-\mathrm{i}H\psi_t^{Q,P}
   +L\widehat z_t^*\psi_t^{Q,P}
   -L^\dagger D_t^Q\psi_t^{Q,P}\bigr)\\
 &=\sum_{i=0}^P(I-\Pi_P)L\widehat z_t^*\psi_{t,i}^{Q,P}\\
 &=L\widehat z_t^*\psi_{t,P}^{Q,P}.
 \end{aligned}
 \label{eq:polynomial-boundary-residual}
\end{equation}
The variation-of-constants formula in Assumption~\ref{ass:convergence}
then yields
\[
 e_t=\mathcal U_Q(t,0)e_0+
 \int_0^t\mathcal U_Q(t,s)R_s\,\dd s.
\]
Since $e_0=0$, unitarity gives
\mbox{$\norm{e_t}_{\mathcal X}
\leq\int_0^t\norm{R_s}_{\mathcal X}\,\dd s$}
\cite[Section~II.1.6, Theorem~1.5]{lubich2008}.

By Eq.~\eqref{eq:polynomial-boundary-residual},
\mbox{$\norm{R_s}_{\mathcal X}
\leq\norm{L}\norm{\widehat z_s^*\psi_{s,P}^{Q,P}}_{\mathcal X}$}.
Set $\beta_Q(s)=\norm{(b_q(s))_{q=1}^Q}_2
=\sqrt{\widehat\alpha(s,s)}$.

The degree-$P$ component is
\mbox{$\psi_{s,P}^{Q,P}=\sum_{|\bm k|=P}
u_{\bm k}(s)\mathcal P_{\bm k}$}.
Using \mbox{$\widehat z_s^*=\sum_{q=1}^Qb_q^*(s)\eta_q^*$}
and Eq.~\eqref{eq:ladder-rules}, we obtain
\begin{equation}
 \begin{aligned}
 \widehat z_s^*\psi_{s,P}^{Q,P}
 &=\left(\sum_{q=1}^Qb_q^*(s)\eta_q^*\right)
 \left(\sum_{|\bm k|=P}u_{\bm k}(s)\mathcal P_{\bm k}\right)\\
 &=\sum_{|\bm k|=P}\sum_{q=1}^Q
 b_q^*(s)u_{\bm k}(s)\eta_q^*\mathcal P_{\bm k}\\
 &=\sum_{|\bm k|=P}\sum_{q=1}^Q
 b_q^*(s)\sqrt{k_q+1}\,u_{\bm k}(s)
 \mathcal P_{\bm k+\bm e_q}\\
 &=\sum_{|\bm\ell|=P+1}
 \left(\sum_{q:\ell_q>0}\sqrt{\ell_q}\,b_q^*(s)
 u_{\bm\ell-\bm e_q}(s)\right)\mathcal P_{\bm\ell}.
 \end{aligned}
 \label{eq:polynomial-multiplication-expansion}
\end{equation}
The last line groups all terms multiplying the same basis function
$\mathcal P_{\bm\ell}$, with $\bm\ell=\bm k+\bm e_q$.
Orthonormality and Cauchy--Schwarz with weights $\sqrt{\ell_q}$ yield
\begin{equation}
 \begin{aligned}
 \norm{\widehat z_s^*\psi_{s,P}^{Q,P}}_{\mathcal X}^2
 &=\sum_{|\bm\ell|=P+1}
 \norm{\sum_{q:\ell_q>0}\sqrt{\ell_q}\,b_q^*(s)
 u_{\bm\ell-\bm e_q}(s)}_{\mathcal H}^2\\
 &\leq\sum_{|\bm\ell|=P+1}
 \left(\sum_{q:\ell_q>0}\ell_q\right)
 \sum_{q:\ell_q>0}|b_q(s)|^2
 \norm{u_{\bm\ell-\bm e_q}(s)}_{\mathcal H}^2\\
 &=(P+1)\sum_{q=1}^Q|b_q(s)|^2
 \sum_{|\bm k|=P}\norm{u_{\bm k}(s)}_{\mathcal H}^2\\
 &=(P+1)\beta_Q(s)^2\norm{\psi_{s,P}^{Q,P}}_{\mathcal X}^2.
 \end{aligned}
 \label{eq:polynomial-multiplication-norm}
\end{equation}
Here we used $\sum_q\ell_q=P+1$ and reindexed
$\bm k=\bm\ell-\bm e_q$ for each $q$.
Taking square roots in Eq.~\eqref{eq:polynomial-multiplication-norm} gives
\[
 \norm{\widehat z_s^*\psi_{s,P}^{Q,P}}_{\mathcal X}
 \leq\sqrt{P+1}\,\beta_Q(s)\norm{\psi_{s,P}^{Q,P}}_{\mathcal X}.
\]
Substituting these bounds into the error integral yields
\begin{equation}
 \norm{e_t}_{\mathcal X}
 \leq\norm{L}\sqrt{P+1}\int_0^t
 \beta_Q(s)\norm{\psi_{s,P}^{Q,P}}_{\mathcal X}\,\dd s.
 \label{eq:polynomial-residual-bound}
\end{equation}

Set $v_i(t)=\psi_{t,i}^{Q,P}$ for $0\leq i\leq P$,
with $v_{-1}=v_{P+1}=0$.
To bound $\norm{v_P}_{\mathcal X}$, project
Eq.~\eqref{eq:polynomial-galerkin-evolution} onto degree $i$:
\begin{equation}
 \partial_tv_i=-\mathrm{i}Hv_i
 +L\widehat z_t^*v_{i-1}
 -L^\dagger D_t^Qv_{i+1},\qquad 0\leq i\leq P.
 \label{eq:polynomial-degree-evolution}
\end{equation}
The multiplication estimate in Eq.~\eqref{eq:polynomial-multiplication-norm},
applied at degree $i-1$, gives
$\norm{\widehat z_t^*v_{i-1}}_{\mathcal X}
\leq\sqrt i\,\beta_Q(t)\norm{v_{i-1}}_{\mathcal X}$.
For the differentiation term, Eq.~\eqref{eq:ladder-rules} gives
\begin{equation}
 \norm{D_t^Qv_{i+1}}_{\mathcal X}
 \leq\beta_Q(t)
 \norm{(\partial_{\eta_q^*}v_{i+1})_{q=1}^Q}_{\ell^2(\mathcal X)}
 =\sqrt{i+1}\,\beta_Q(t)\norm{v_{i+1}}_{\mathcal X}.
 \label{eq:polynomial-differentiation-bound}
\end{equation}
Variation of constants in Eq.~\eqref{eq:polynomial-degree-evolution},
unitarity of $e^{-\mathrm{i}(t-s)H}$, and
$v_i(0)=\delta_{i0}\psi_0$ therefore yield
\begin{equation}
 \begin{aligned}
 \norm{v_i(t)}_{\mathcal X}
 &\leq\delta_{i0}+\norm{L}\int_0^t
 \left(\norm{\widehat z_s^*v_{i-1}(s)}_{\mathcal X}
 +\norm{D_s^Qv_{i+1}(s)}_{\mathcal X}\right)\,\dd s\\
 &\leq\delta_{i0}+\norm{L}\int_0^t\beta_Q(s)
 \left(\sqrt i\,\norm{v_{i-1}(s)}_{\mathcal X}
 +\sqrt{i+1}\,\norm{v_{i+1}(s)}_{\mathcal X}\right)\,\dd s.
 \end{aligned}
 \label{eq:polynomial-degree-comparison}
\end{equation}

Set $A_Q(t)=\norm{L}\int_0^t\beta_Q(s)\,\dd s$.
To construct an explicit upper bound in
Eq.~\eqref{eq:polynomial-degree-comparison}, define
\begin{equation}
 w_i(t):=e^{A_Q(t)^2/2}\frac{A_Q(t)^i}{\sqrt{i!}},
 \qquad i\geq0,
 \label{eq:polynomial-comparison-functions}
\end{equation}
with $w_{-1}=0$. Direct differentiation gives
\begin{equation}
 w_i'(t)=\norm{L}\beta_Q(t)
 \left(\sqrt i\,w_{i-1}(t)+\sqrt{i+1}\,w_{i+1}(t)\right),
 \qquad w_i(0)=\delta_{i0}.
 \label{eq:polynomial-comparison-evolution}
\end{equation}
Integrating from $0$ to $t$ gives
\begin{equation}
 w_i(t)=\delta_{i0}+\norm{L}\int_0^t\beta_Q(s)
 \left(\sqrt i\,w_{i-1}(s)+\sqrt{i+1}\,w_{i+1}(s)\right)\,\dd s.
 \label{eq:polynomial-comparison-integral}
\end{equation}
To compare $\norm{v_i}_{\mathcal X}$ and $w_i$, set
$d(t)=\max_{0\leq i\leq P}(\norm{v_i(t)}_{\mathcal X}-w_i(t))_+$,
where $x_+=\max\{x,0\}$.
Subtracting Eq.~\eqref{eq:polynomial-comparison-integral} from
Eq.~\eqref{eq:polynomial-degree-comparison}, we bound each neighboring
difference by $d(s)$, including the boundary terms since
$v_{-1}=w_{-1}=0$ and $\norm{v_{P+1}}_{\mathcal X}=0\leq w_{P+1}$.
Taking positive parts and the maximum over $0\leq i\leq P$, and using
$\sqrt i+\sqrt{i+1}\leq2\sqrt{P+1}$, gives
\begin{equation}
 d(t)\leq2\sqrt{P+1}\,\norm{L}
 \int_0^t\beta_Q(s)d(s)\,\dd s.
 \label{eq:polynomial-comparison-gap}
\end{equation}
Gr\"onwall's inequality yields $d=0$, so
\begin{equation}
 \norm{v_i(t)}_{\mathcal X}\leq w_i(t)
 =e^{A_Q(t)^2/2}\frac{A_Q(t)^i}{\sqrt{i!}},
 \qquad 0\leq i\leq P.
 \label{eq:polynomial-degree-majorant}
\end{equation}
Substituting Eq.~\eqref{eq:polynomial-degree-majorant} with $i=P$
into Eq.~\eqref{eq:polynomial-residual-bound} and using
$A_Q'(s)=\norm{L}\beta_Q(s)$, we obtain
\begin{equation}
 \begin{aligned}
 \norm{e_t}_{\mathcal X}
 &\leq\frac{\sqrt{P+1}}{\sqrt{P!}}
       \int_0^{A_Q(t)}e^{a^2/2}a^P\,\dd a\\
 &\leq e^{A_Q(t)^2/2}
       \frac{A_Q(t)^{P+1}}{\sqrt{(P+1)!}}
 \leq e^{A_T/2}\frac{A_T^{(P+1)/2}}{\sqrt{(P+1)!}}.
 \end{aligned}
 \label{eq:polynomial-final-error-bound}
\end{equation}
Here Cauchy--Schwarz and $\widehat\alpha(s,s)\leq\alpha(s,s)$ give
\mbox{$A_Q(t)^2\leq T\norm{L}^2\int_0^T\widehat\alpha(s,s)\,\dd s
\leq A_T$} for $t\leq T$.
Squaring and taking the supremum in time proves the result.
\end{proof}

\subsection{Numerical coefficient error and the total bound}
\label{sec:numerical-solver-error}

\begin{assumption}[Numerical accuracy]
\label{ass:numerical-accuracy}
For fixed $Q$, $P$, and $T$, the numerical coefficients satisfy
\begin{equation}
 \max_{0\leq n\leq N_T}
 \left(\sum_{\bm k\in\mathcal K_{Q,P}}
 \norm{u_{\bm k}(\cdot,t_n)-\widehat u_{\bm k}^n(\cdot)}_{\mathcal H}^2
 \right)^{1/2}
 \leq\delta_{\mathrm{num}},
 \label{eq:numerical-coefficient-accuracy}
\end{equation}
for a numerical error bound $\delta_{\mathrm{num}}\geq0$.
\end{assumption}

The bound $\delta_{\mathrm{num}}$ accounts for the numerical construction
of the kernel factors and for spatial and time discretization.

By orthonormality of the stochastic basis $\mathcal P_{\bm k}$,
\begin{equation}
 \norm{\psi_{t_n}^{Q,P}-\psi^{Q,P,n}}_{\mathcal X}
 =\left(\sum_{\bm k\in\mathcal K_{Q,P}}
 \norm{u_{\bm k}(\cdot,t_n)-\widehat u_{\bm k}^n(\cdot)}_{\mathcal H}^2
 \right)^{1/2}
 \leq\delta_{\mathrm{num}}.
 \label{eq:numerical-solver-error}
\end{equation}

\begin{proof}[Proof of Theorem~\ref{thm:main-error}]
At each time node,
\begin{equation}
 \begin{aligned}
 \norm{\psi_{t_n}-\psi^{Q,P,n}}_{\mathcal X}
 \leq{}&\norm{\psi_{t_n}-\psi_{t_n}^Q}_{\mathcal X}\\
 &+\norm{\psi_{t_n}^Q-\psi_{t_n}^{Q,P}}_{\mathcal X}
 +\norm{\psi_{t_n}^{Q,P}-\psi^{Q,P,n}}_{\mathcal X}.
 \end{aligned}
 \label{eq:total-error-decomposition}
\end{equation}
Apply Theorems~\ref{thm:kl-error} and~\ref{thm:chaos-error} and
Eq.~\eqref{eq:numerical-solver-error}, then take the maximum over $n$.
Choosing $C_2=\sqrt{A_T}$ and
$C_1=\max\{\sqrt T\norm{L},e^{A_T/2}\}$ gives the result.
\end{proof}

\subsection{Trace conservation}
\label{app:trace-conservation}

The Hamiltonian contribution vanishes because
$\langle u_{\bm k},Hu_{\bm k}\rangle_{\mathcal H}$ is real,
while the coupling terms cancel in pairs.
For each $\bm k,\bm k+\bm e_q\in\mathcal K_{Q,P}$,
\begin{equation}
 2\sqrt{k_q+1}\,\operatorname{Re}\!\left[
 b_q^*(t)\langle u_{\bm k+\bm e_q},Lu_{\bm k}\rangle_{\mathcal H}
 -b_q(t)\langle u_{\bm k},L^\dagger u_{\bm k+\bm e_q}\rangle_{\mathcal H}
 \right]=0.
 \label{eq:coupling-norm-cancellation}
\end{equation}
The two products in brackets are complex conjugates, so their difference
has zero real part. Summing over all retained indices therefore gives
\begin{equation}
 \frac{\dd}{\dd t}\sum_{\bm k\in\mathcal K_{Q,P}}
 \norm{u_{\bm k}}^2
 =2\operatorname{Re}\sum_{\bm k\in\mathcal K_{Q,P}}
 \langle u_{\bm k},\partial_tu_{\bm k}\rangle_{\mathcal H}=0.
 \label{eq:coefficient-norm-conservation}
\end{equation}
By Eq.~\eqref{eq:chaos-density},
$\Tr\rho_m(t)=\sum_{\bm k\in\mathcal K_{Q,P}}\norm{u_{\bm k}}^2$,
which proves Eq.~\eqref{eq:mass-conservation}.
The initial conditions give $\Tr\rho_m(t)=\norm{\psi_0}^2=1$.
The same cancellation holds after an orthogonal Galerkin projection
in $\mathcal H$.

\section{Multidimensional Caldeira--Leggett model and discretization}
\label{app:multidimensional-implementation}

\subsection{Model and coefficient equations}

For the multidimensional experiments, let $\bm x=(x_1,\ldots,x_d)\in\mathbb R^d$.
We attach an independent copy of
the oscillator bath to each coordinate direction.  With the same bath
spectrum in every direction, the complete model is
\begin{equation}
 \begin{aligned}
 H_{\mathrm{tot}}^{(d)}={}&H_{\mathrm{sys}}^{(d)}\otimes I_{\mathrm{bath}}
 +I_{\mathrm{sys}}\otimes H_{\mathrm{bath}}^{(d)}
 +H_{\mathrm{int}}^{(d)},\\
 H_{\mathrm{sys}}^{(d)}={}&-\frac{\varepsilon^2}{2}\Delta_{\bm x}
 +V(\bm x),
 &V(\bm x)={}&V_0(\bm x)
 +\frac{\kappa}{2}|\bm x|^2,\\
 H_{\mathrm{bath}}^{(d)}={}&
 \sum_{\mu=1}^{d}\sum_{j=1}^{N_{\mathrm{b}}}
 \left(-\frac{\varepsilon^2}{2}\partial_{r_{j,\mu}}^2
 +\frac{\omega_j^2r_{j,\mu}^2}{2}\right),\\
 H_{\mathrm{int}}^{(d)}={}&
 \sum_{\mu=1}^{d}W_{\mathrm{s},\mu}\otimes W_{\mathrm{b},\mu},
 &(W_{\mathrm{s},\mu}f)(\bm x)={}&x_\mu f(\bm x),\\
 W_{\mathrm{b},\mu}={}&-\sum_{j=1}^{N_{\mathrm{b}}}c_jr_{j,\mu}.
 \end{aligned}
 \label{eq:cl-model-d}
\end{equation}

The NMQSD quantities corresponding to Eq.~\eqref{eq:cl-model-d} are
\begin{equation}
 H=\frac{H_{\mathrm{sys}}^{(d)}}{\varepsilon},
 \qquad
 L_\mu=W_{\mathrm{s},\mu},
 \qquad
 \alpha_{\mu\nu}(t,s)=\delta_{\mu\nu}B_\varepsilon(t,s).
 \label{eq:nmqsd-data-d}
\end{equation}
The corresponding $d$-channel equation is
\begin{equation}
 \partial_t\psi(\bm x,t;\bm z)=\left[-iH+
 \sum_{\mu=1}^{d}\left(
 L_\mu z_{\mu,t}^*-L_\mu^\dagger\int_0^t
 B_\varepsilon(t,s)\frac{\delta}{\delta z_{\mu,s}^*}\,\dd s
 \right)\right]\psi(\bm x,t;\bm z),
 \label{eq:nmqsd-d}
\end{equation}
where
\begin{equation}
 \E[z_{\mu,t}z_{\nu,s}]=0,
 \qquad
 \E[z_{\mu,t}z_{\nu,s}^*]
 =\delta_{\mu\nu}B_\varepsilon(t,s).
 \label{eq:process-d}
\end{equation}

For each bath, the low-rank approximation in
Section~\ref{sec:low-rank-modes} gives
\begin{equation}
 z_{\mu,t}\approx\sum_{q=1}^{Q}b_{\mu q}(t)\eta_{\mu q},
 \qquad \mu=1,\ldots,d,
 \label{eq:multichannel-mode-representation}
\end{equation}
where $b_{\mu q}$ are the retained modes of $B_\varepsilon$ and all
$\eta_{\mu q}\sim\CN(0,1)$ are independent. The polynomial Galerkin
projection uses the joint total-degree index set
\begin{equation}
 \mathcal K_{dQ,P}=\left\{\bm k\in\mathbb N_0^{dQ}:
 \sum_{\mu=1}^{d}\sum_{q=1}^{Q}k_{\mu q}\leq P\right\}.
 \label{eq:multichannel-total-degree}
\end{equation}
Here $k_{\mu q}$ and $\bm e_{\mu q}$ denote the component and unit
multi-index associated with bath $\mu$ and mode $q$.

For spatial discretization, use the product basis
\begin{equation}
 \chi_{\bm n}(\bm x)=\prod_{\mu=1}^d\chi_{n_\mu}^{(\mu)}(x_\mu),
 \qquad \bm n\in\mathcal I_{\mathrm{sp}},\qquad
 K=|\mathcal I_{\mathrm{sp}}|.
 \label{eq:spatial-product-basis}
\end{equation}
Here each $\{\chi_n^{(\mu)}\}_{n\geq0}$ is orthonormal and
$\mathcal I_{\mathrm{sp}}\subset\mathbb N_0^d$ is finite.
The cutoffs can be chosen as $0\leq n_\mu<K_\mu$ for each $\mu$,
or $|\bm n|\leq\ell_{\max}$.

After enumerating the spatial basis, set
$(H_K)_{ij}=\langle\chi_i,H\chi_j\rangle$ and
$(L_{\mu,K})_{ij}=\langle\chi_i,L_\mu\chi_j\rangle$.
The multidimensional coefficient equations are
\begin{equation}
 \begin{aligned}
 \partial_t\bm c_{\bm k}
 ={}&-iH_K\bm c_{\bm k}\\
 &+\sum_{\mu=1}^d\sum_{q=1}^Q
 \left[
 b_{\mu q}^*(t)\sqrt{k_{\mu q}}\,
 L_{\mu,K}\bm c_{\bm k-\bm e_{\mu q}}
 -b_{\mu q}(t)\sqrt{k_{\mu q}+1}\,
 L_{\mu,K}^\dagger\bm c_{\bm k+\bm e_{\mu q}}
 \right].
 \end{aligned}
 \label{eq:multichannel-coefficient-coupling}
\end{equation}
Coefficients outside $\mathcal K_{dQ,P}$ vanish.

\subsection{Gaussian reference solution}
\label{app:gaussian-reference}

For Example~\ref{ex:cl-coupled-3d}, the reference density
\eqref{eq:three-dimensional-exact-density} is obtained from the Gaussian
dynamics of the full quadratic Hamiltonian~\cite{babakan2026}.
Collect all positions in
$\bm y=(\bm x^T,\bm r_1^T,\bm r_2^T,\bm r_3^T)^T$, where
$\bm r_\mu=(r_{1,\mu},\ldots,r_{N_{\mathrm b},\mu})^T$, and let
$\bm p=-i\varepsilon\nabla_{\bm y}$.
Then Eq.~\eqref{eq:cl-model-d} becomes
\begin{equation}
 \begin{aligned}
 H_{\mathrm{tot}}^{(3)}&=\frac12\bm p^T\bm p
 +\frac12\bm y^T\mathcal A\bm y,\\
 \mathcal A&=
 \begin{pmatrix}
 A+\kappa I_3 & -I_3\otimes\bm c^T\\
 -I_3\otimes\bm c & I_3\otimes\Omega^2
 \end{pmatrix},
 \end{aligned}
 \label{eq:three-dimensional-full-potential}
\end{equation}
where $A$ is given in Eq.~\eqref{eq:three-dimensional-potential},
$\Omega=\operatorname{diag}(\omega_1,\ldots,\omega_{N_{\mathrm b}})$,
and $\bm c=(c_1,\ldots,c_{N_{\mathrm b}})^T$.
The counterterm $\kappa$ is defined in Eq.~\eqref{eq:effective-potential}.

For the initial state~\eqref{eq:three-dimensional-initial} and independent
thermal baths, the initial position mean and the position and momentum
covariance matrices are
\begin{equation}
 \begin{aligned}
 \bm m_0&=(\bm a^T,\bm 0^T)^T,\\
 \Gamma_y&=\frac{\varepsilon}{2}\operatorname{diag}
 \bigl(D_0^{-1},I_3\otimes(\Omega^{-1}\Theta)\bigr),\\
 \Gamma_p&=\frac{\varepsilon}{2}\operatorname{diag}
 \bigl(D_0,I_3\otimes(\Omega\Theta)\bigr),\\
 \Theta&=\operatorname{diag}_{1\leq j\leq N_{\mathrm b}}
 \left[\coth\!\left(\frac{\beta\varepsilon\omega_j}{2}\right)\right].
 \end{aligned}
 \label{eq:three-dimensional-reference-initial-moments}
\end{equation}
The initial momentum mean and symmetrized position--momentum covariance
vanish. With
\[
 U(t)=\cos(t\mathcal A^{1/2}),\qquad
 W(t)=\mathcal A^{-1/2}\sin(t\mathcal A^{1/2}),
\]
the Heisenberg evolution is $\bm y(t)=U(t)\bm y(0)+W(t)\bm p(0)$.
Projecting onto the system coordinates with $E=(I_3\ \ 0)$ gives
\begin{equation}
 \begin{aligned}
 \bm\mu(t)&=EU(t)\bm m_0,\\
 \Sigma_x(t)&=E\bigl[U(t)\Gamma_yU(t)^T
 +W(t)\Gamma_pW(t)^T\bigr]E^T.
 \end{aligned}
 \label{eq:three-dimensional-reference-moments}
\end{equation}

\bibliography{references}

\end{document}